\documentclass[12pt,draftcls,peerreviewca,onecolumn]{IEEEtran}
\usepackage{etoolbox}
\makeatletter
\patchcmd{\@makecaption}
  {\scshape}
  {}
  {}
  {}
\makeatletter
\patchcmd{\@makecaption}
  {\\}
  {.\ }
  {}
  {}
\makeatother
\usepackage{amsfonts}
\usepackage{mathrsfs}
\usepackage{amsfonts}
\usepackage{amssymb}
\usepackage{graphicx}
\usepackage{epsfig}
\usepackage{psfrag}
\usepackage{amsmath}
\usepackage{array}
\usepackage{cases}
\usepackage{subfigure}
\usepackage{cite,graphicx,amsmath,amssymb,color}
\usepackage{anysize}
\usepackage{algorithmic}
\usepackage{algorithm}

\newcommand{\dis}{\stackrel{d}{\sim}}
\newcommand{\eqla}{\stackrel{(a)}{=}}
\newcommand{\eqlb}{\stackrel{(b)}{=}}
\newcommand{\eqlc}{\stackrel{(c)}{=}}
\newcommand{\eqld}{\stackrel{(d)}{=}}
\newcommand{\eqle}{\stackrel{(e)}{=}}

\newtheorem{Thm}{Theorem}
\newtheorem{Lem}{Lemma}
\newtheorem{Cor}{Corollary}

\newtheorem{Prob}{Problem}
\newtheorem{Rem}{Remark}

\IEEEoverridecommandlockouts

\begin{document}

\title{Analysis and Optimization of  Caching and Multicasting  in Large-Scale  Cache-Enabled Wireless Networks}
\author{Ying Cui\thanks{Y. Cui and D. Jiang are with the Department of  Electronic Engineering, Shanghai Jiao Tong University, China.  Y. Wu is with the Department of Electrical and Electronic Engineering, The University of Hong Kong, Hong Kong. This paper was  presented in part at IEEE GLOBECOM 2015 \cite{Cui15SGOPTGC}. Y. Cui was supported in part by the National Science Foundation of China
grant 61401272.},  {\em MIEEE}, \and Dongdong Jiang, {\em StMIEEE}, \and Yueping Wu,  {\em MIEEE}}

\maketitle

\begin{abstract}
Caching and multicasting at base stations are two
promising approaches to support massive content delivery over
wireless networks. However, existing analysis and designs do not fully explore and exploit
the potential advantages of the two approaches.  In this paper, we consider  the analysis and optimization of caching and multicasting  in a large-scale  cache-enabled wireless  network. We propose a random caching and multicasting scheme with a design parameter.
By carefully handling different types of interferers and adopting appropriate approximations,
we derive a tractable  expression  for the successful transmission probability in the general   region,  utilizing tools from stochastic geometry. We also obtain  a closed-form expression for the successful transmission probability
in the high signal-to-noise ratio (SNR) and user density region. Then,  we consider the successful transmission probability maximization, which is a very complex non-convex problem in general.
Using optimization techniques,  we develop an iterative  numerical algorithm to obtain a local optimal caching and multicasting design in the general region.
To reduce complexity and maintain superior performance, we also derive an asymptotically optimal  caching and multicasting  design  in the asymptotic region,  based on a two-step optimization framework.
 Finally, numerical simulations  show that the asymptotically  optimal design
 achieves  a significant gain in successful transmission probability over some baseline schemes in the general  region.
\end{abstract}


\section{Introduction}

The demand for wireless communication services has been
shifting from connection-oriented services such as traditional
voice telephony and messaging to content-oriented services
such as multimedia, social networking and smartphone
applications.
Recently, to support the dramatic growth of  wireless data traffic, caching at base
stations (BSs) has been proposed as a promising approach
for massive content delivery by reducing the distance between popular contents and their requesters (i.e., users) \cite{CMag14Chen,Bastug14,cachingmimoLiu15}. In \cite{TassiulasIT13,Shanmugam13, Poularakis14},  optimal caching designs are considered for simple models of  wireless networks, 
where channel fading is not considered.
For example, in \cite{TassiulasIT13}, the authors consider the optimal caching and delivery design to minimize the link capacity required to sustain a certain request process in a 2-D square grid network.
The  combinatorial optimization problem  is reduced to a simpler continuous optimization problem with the same order of performance as the original problem. 
In \cite{Shanmugam13}, the authors consider the optimal caching design  to minimize the expected file downloading time   in a helper-enabled cellular network   modeled by a  bipartite graph.
The optimization problem for the uncoded case is NP-hard, and a greedy caching design is proposed to achieve performance within a factor 2 of the optimum. 
In \cite{Poularakis14}, the authors consider the optimal caching and routing design in a small cell network.
The optimization problem is NP-hard, and a novel reduction to a variant
of the facility location problem is proposed to obtain a caching design with approximation
guarantees.


In \cite{EURASIP15Debbah,ICC15Giovanidis,Malak14,DBLP:journals/corr/Tamoor-ul-Hassan15}, more realistic network models based on stochastic geometry are considered to characterize the stochastic natures of  channel fading and geographic locations of  BSs and users. Based on these models,  the performance of caching designs at BSs is  analyzed and optimized. Specifically, in \cite{EURASIP15Debbah}, the authors  consider identical caching  at BSs (i.e., all BSs store the same set of popular files), and analyze the outage probability and average rate.
In \cite{ICC15Giovanidis}, the authors consider  random caching at BSs, and analyze and optimize  the hit probability. 
In \cite{Malak14}, the authors consider random caching in D2D networks, and analyze and optimize the total coverage probability, assuming that content  requests follow Zipf distribution.
In \cite{DBLP:journals/corr/Tamoor-ul-Hassan15}, the authors consider random caching of a uniform distribution, and analyze the cache hit probability and the content outage probability, assuming that content  requests follow  a uniform distribution. In addition, in \cite{Altman13} and \cite{DBLP:journals/corr/BharathN15}, the authors also model  cache-enabled wireless networks using stochastic geometry without capturing channel fading (in the performance metrics). Specifically, \cite{Altman13} studies   the expected costs of obtaining a complete content under random uncoded caching and coded caching strategies, which are designed only for different pieces of a single content. Reference \cite{DBLP:journals/corr/BharathN15} considers random caching with contents being stored at each BS in an i.i.d. manner,  and analyzes the minimum offloading loss, assuming unknown file popularity profiles. 
Note that the identical caching design in \cite{EURASIP15Debbah} cannot provide file diversity at different BSs, and hence may not sufficiently exploit  storage resources. In contrast,  the random caching designs in \cite{ICC15Giovanidis,Malak14,DBLP:journals/corr/Tamoor-ul-Hassan15,Altman13,DBLP:journals/corr/BharathN15} can provide file diversity. However, the i.i.d. random caching   in \cite{DBLP:journals/corr/BharathN15} may still waste  storage resources, as multiple copies of a file may be stored at each BS. The random caching designs in \cite{Malak14,DBLP:journals/corr/Tamoor-ul-Hassan15,Altman13} only provide caching probabilities of files and do not specify how  multiple different files can be efficiently stored at each BS based on these probabilities. In
\cite{ICC15Giovanidis}, the authors provide a graphical method to obtain caching probabilities of file combinations based on the caching probabilities of files, which may not be systematic for   large-scale networks.
In addition,  in \cite{EURASIP15Debbah,ICC15Giovanidis,Malak14,DBLP:journals/corr/Tamoor-ul-Hassan15,Altman13,DBLP:journals/corr/BharathN15},  transmission schemes for  cached files are not specified. The performance metrics are only related to the coverage probability, and do not reflect resource sharing among users.

On the other hand, enabling multicast service at BSs is  an efficient way to deliver popular contents to multiple
requesters simultaneously by effectively utilizing the broadcast nature of the wireless medium\cite{eMBMS}.
Caching and multicasting at BSs  have been  jointly considered in the
literature. For example, in \cite{WCNC14Tassiulas}, the authors propose a heuristic caching algorithm for a given multicasting design
 to reduce the service cost.  In \cite{TON14Shakkottai}, the authors  propose a joint throughput-optimal caching and multicasting
algorithm   to maximize the service rate.  In \cite{ZhouISIT15}, the authors obtain
the optimal dynamic multicast scheduling to
  minimize the system cost. Note that \cite{WCNC14Tassiulas,TON14Shakkottai,ZhouISIT15} focus on algorithm design instead of performance analysis, and do not offer many insights into the fundamental impact of caching and multicasting.
 In addition, \cite{WCNC14Tassiulas,TON14Shakkottai,ZhouISIT15}  consider  simple network models which cannot capture the geographic  features of the locations of BSs and users or statistical properties of signal and interference.

In summary, further studies are required to understand the impact of caching and multicasting   in cache-enabled wireless  networks. In this paper, we  consider the analysis and optimization of joint caching and multicasting in a  large-scale wireless  network. Our network model effectively characterizes  the stochastic natures of
channel fading and geographic locations of  BSs and users. The main contributions of this paper are summarized below.

$\bullet$ First, we propose a random caching and multicasting scheme with a design parameter to effectively  improve the efficiency of information dissemination.  Different from \cite{EURASIP15Debbah,ICC15Giovanidis,Malak14,DBLP:journals/corr/Tamoor-ul-Hassan15,Altman13,DBLP:journals/corr/BharathN15}, we specify a transmission scheme, i.e., multicasting, which exploits  the broadcast nature of the wireless medium. In addition, different from the random caching designs in \cite{ICC15Giovanidis,Malak14,DBLP:journals/corr/Tamoor-ul-Hassan15,Altman13,DBLP:journals/corr/BharathN15}, the proposed random caching design is on the basis of file combinations, and can make better use of storage resources when multicasting is considered.

$\bullet$ Next, we analyze the successful transmission probability. Caching on the basis of file combinations significantly complicates the distributions of user associations and interferers, and hence makes the analysis very challenging. By carefully handling different types of interferers and adopting appropriate approximations, we derive a tractable  expression  for the successful transmission probability in the general  region, utilizing tools from stochastic geometry. We also derive  a closed-form expression  for the successful transmission probability in  the high SNR and user density region.\footnote{Note that the tractable expression  involves integrals, but can be easily computed using numerical computation tools, e.g., MATLAB. The closed-form expression does not have unsolvable integrals, and  can be computed much more efficiently.
} These expressions reveal the impacts of physical layer   and information-related parameters on the successful transmission probability.

$\bullet$ Then, we consider the successful transmission probability maximization, which is a very complex non-convex problem in general, due to the sophisticated structure of the  successful transmission probability.
Using optimization techniques, we develop an iterative  numerical algorithm to obtain a local optimal caching and multicasting  design in the general  region. To reduce complexity and maintain superior performance, we derive an asymptotically optimal  caching and multicasting  design  in the high SNR and user density region, based on a two-step optimization framework.
The features of the asymptotically optimal design  provide  important design insights.

$\bullet$ Finally, by numerical simulations, we show that the asymptotically  optimal design has much lower computational complexity than the local optimal design and achieves a significant  gain in successful transmission probability in the general  region   over some baseline schemes.

\section{System Model}

\subsection{Network Model}\label{sec:netmodel}
We consider a  large-scale   network, as shown in Fig.~\ref{fig:system}. The locations of the base stations (BSs) are spatially distributed as a homogeneous Poisson point process (PPP) $\Phi_{b}$ with density $\lambda_{b}$. The locations of the users are distributed as an independent homogeneous PPP with density $\lambda_{u}$.
We consider the downlink transmission. Each BS has one transmit antenna with   transmission power $P$. Each user has one receive antenna. The total bandwidth is $W$ (Hz). Consider a discrete-time system with time being slotted and study one slot of the network.
We consider both large-scale fading and small-scale fading. Specifically,
due to large-scale fading,  transmitted
signals with distance $D$ are attenuated by a factor $D^{-\alpha}$, where $\alpha>2$ is the path loss exponent.
For small-scale fading, we assume Rayleigh fading, i.e., each small-scale channel  $h\dis\mathcal {CN}(0, 1)$ or channel power $|h|^2\dis \text{Exp}(1)$\cite{Tse:2005:FWC}.\footnote{Rayleigh fading is a widely used probabilistic channel model in existing literature\cite{Andrews11,WCOM13Andrews}. It is based on the assumption that there are a large number of statistically independent reflected and scattered paths with random amplitudes and random phases uniformly distributed in $[0,2\pi]$, and {\em Central Limit Theorem}\cite{Tse:2005:FWC}.}

Let $\mathcal N\triangleq \{1,2,\cdots, N\}$ denote the set of $N\geq 1$ files (contents) in the network.
For ease of illustration, we assume that all  files  have the same size.\footnote{Note that  files  of different sizes can be divided into chunks  of the same  length. Thus, the results in this paper can be  extended to the case of different file sizes.}  Each file is of certain popularity. We assume that  the file popularity distribution is identical among all users.  Each user randomly  requests one file, which is file $n\in \mathcal N$ with probability $a_n\in (0,1]$, where $\sum_{n\in \mathcal N}a_n=1$.  Thus, the file popularity distribution is given by $\mathbf a\triangleq (a_n)_{n\in \mathcal N }$, which   is assumed to be known   apriori.  In addition, without loss of generality (w.l.o.g.), we assume $a_{1}\ge a_{2}\ldots\ge a_{N}$, i.e., the popularity rank of file $n$ is $n$.

 The   network consists of  cache-enabled BSs. In particular, each BS is equipped with a cache of size $K\leq N$, storing  $K$  different popular files out of $N$, as illustrated in Fig.~\ref{fig:system}.\footnote{Note that   storing more than one copies of the same file at one BS is redundant and  will waste storage resources.}   We say every $K$ different files form a combination. Thus, there are $I\triangleq \binom{N}{K}$ different combinations of $K$  different  files in total. Let $\mathcal I\triangleq \{1,2,\cdots, I\}$ denote the set of $I$ combinations. Combination $i\in \mathcal I$ can be characterized  by an $N$-dimensional vector   $\mathbf x_i\triangleq (x_{i,n})_{n\in \mathcal N}$, where $x_{i,n}=1$ indicates that file $n$ is included in combination $i$ and $x_{i,n}=0$ otherwise. Note that there are  $K$ 1's  in  each $\mathbf x_i$. 
 Denote $\mathcal N_i\triangleq \{n:x_{i,n}=1\}$ as the set of $K$   files contained in combination $i$. Note that in practical networks with large $N$ and $K$, it may not be possible to enumerate all the combinations in $\mathcal I$. However, to understand the natures of joint caching and multicasting in cache-enabled wireless networks, in this paper, we shall  first pose the analysis and optimization on the basis of all the file combinations in $\mathcal I$.  Then, based on the insights obtained, we shall focus on reducing  complexity while maintaining superior performance.

\begin{figure}
\begin{center}
  \subfigure[\small{$K=1$.}]
  {\resizebox{4cm}{!}{\includegraphics{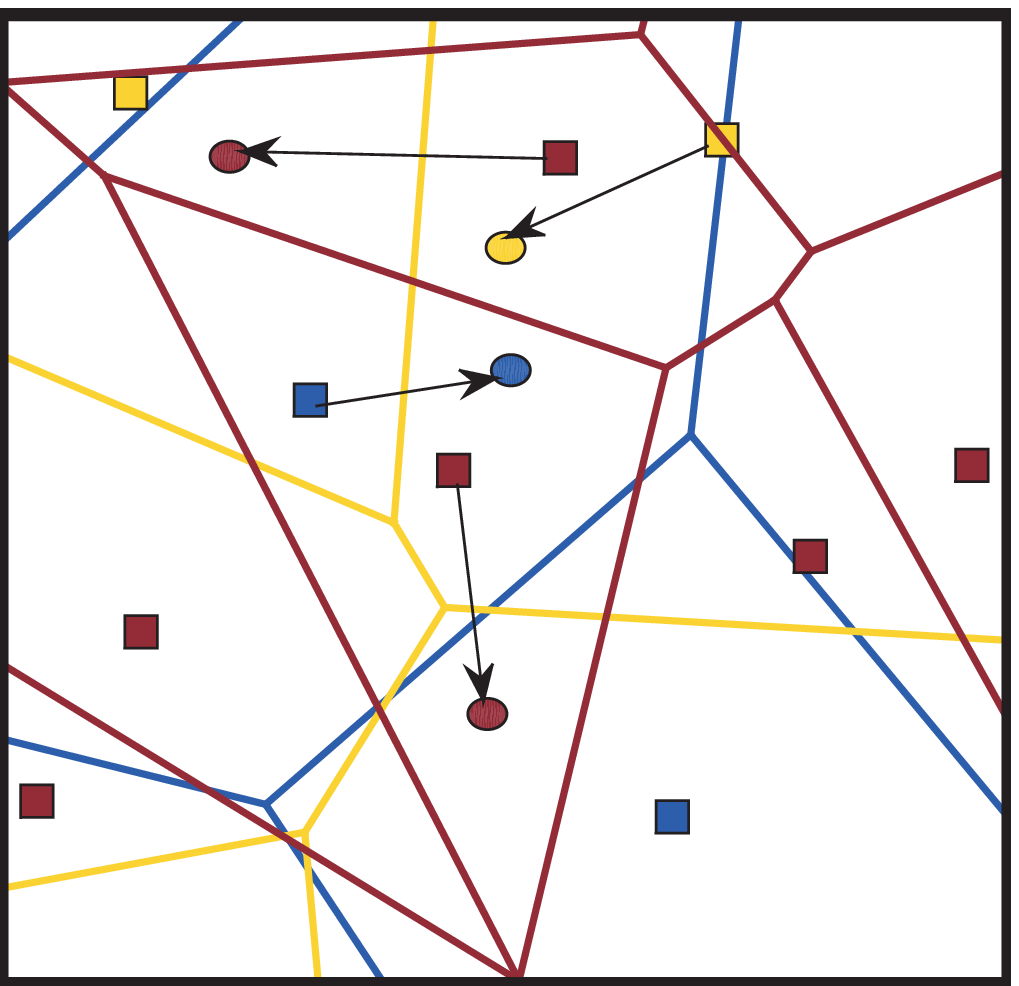}}}\quad
  \subfigure[\small{$K=2$.}]
  {\resizebox{4.3cm}{!}{\includegraphics{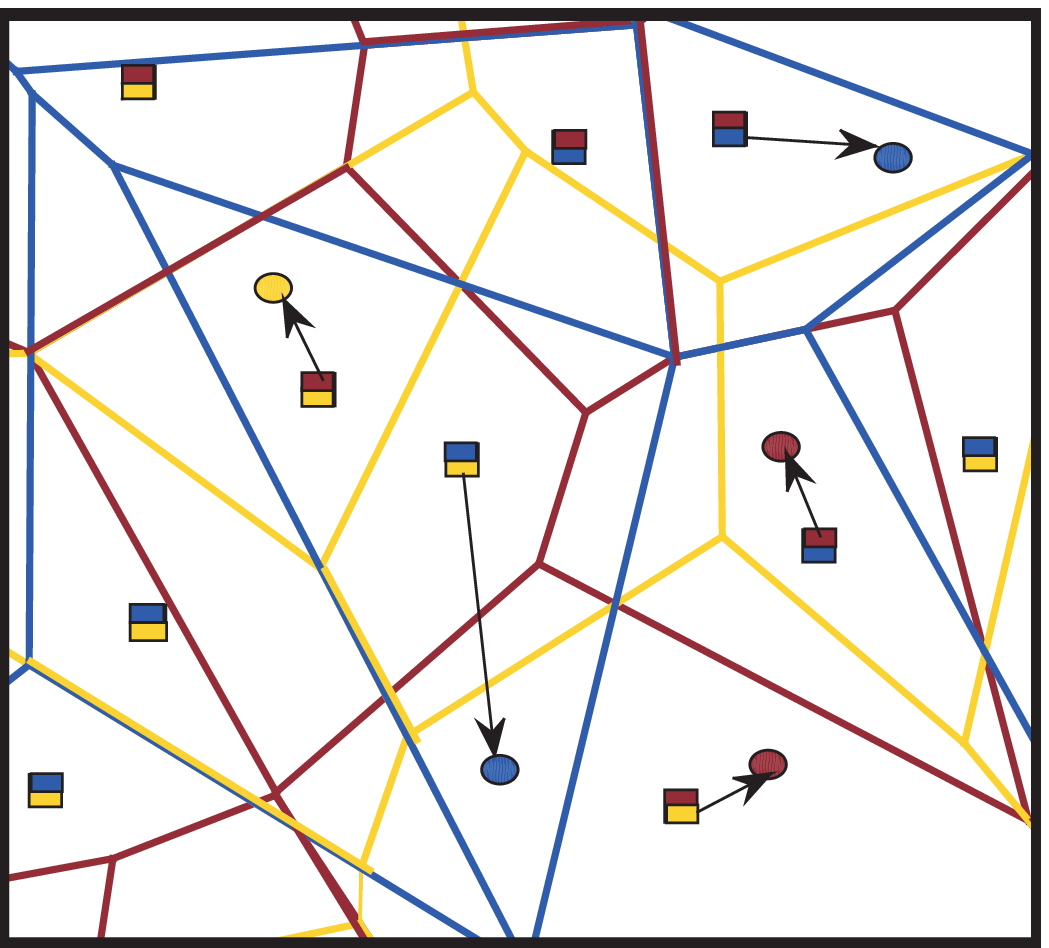}}}
  \end{center}
    \caption{\small{System model. There are three files ($N=3$) in the network, represented by the red, yellow and blue colors, respectively. Each circle represents a user, the color of which indicates the file requested by the user. Each square represents a BS, the color of which indicates the file ($K=1$) or the two different files ($K=2$) stored at the BS.  For each file, there is a corresponding Voronoi tessellation in the same color, which is determined  by the locations of the BSs storing this file. Each arrow represents a signal.}
 }
\label{fig:system}
\end{figure}

\subsection{Caching and Multicasting}\label{subsec:description}

 We consider random caching on the basis of file combinations.
Each BS stores one combination at random, which is combination $i\in \mathcal I$ with probability $p_i$ satisfying
\begin{align}
&0\leq p_i\leq1,\ i\in \mathcal I,\label{eqn:cache-constr-indiv}\\
&\sum_{i\in \mathcal I}p_i=1.\label{eqn:cache-constr-sum}
\end{align}
A random caching design is specified by the caching distribution $\mathbf p\triangleq (p_i)_{ i\in \mathcal I}$.  Let $\mathcal I_n\triangleq \{i\in \mathcal I:x_{i,n}=1\}$ denote the set of $I_n\triangleq \binom{N-1}{K-1} $ combinations  containing file $n$. Let
\begin{align}
T_n\triangleq \sum_{i\in\mathcal{I}_{n}}p_{i}, \  n\in \mathcal N\label{eqn:def-T-n}
\end{align}
denote the probability that file $n$ is stored at a BS. Note that $T_n=p_n$ when $K=1$.
In this paper, we focus on serving  cached files at  BSs to get the first-order insights into the design of cache-enabled wireless networks, as in \cite{ICC15Giovanidis,Malak14}.   BSs may serve uncached files through other service mechanisms.\footnote{For example, BSs can fetch some uncached files from the core network through  backhaul links and multicast them  over other reserved frequency bands. The service of uncached files may involve  backhaul cost or extra delay.} The investigation of service mechanisms for uncached files  is beyond the scope of this paper.

  As illustrated in Fig.~\ref{fig:system}, each user requesting  file $n$ is associated with the nearest BS storing a combination $i\in \mathcal I_n$, referred to as its serving BS, as this BS offers the maximum long-term average receive power for file $n$ \cite{Andrews11,WCOM13Andrews}. Note that the serving BS of a user may not be its  geographically nearest BS and  is also statistically determined by the caching distribution $\mathbf p$.  We refer to this association mechanism as the {\em information-centric association}.  Different from the traditional {\em connection-based association} \cite{Andrews11,WCOM13Andrews}, this association jointly considers the physical layer  and information-centric properties.

We adopt multicast service\footnote{Note that in this paper the multicast service happens once every slot, and hence no additional delay is introduced compared with the traditional unicast scheme.} in the cache-enabled  wireless network for efficient information dissemination.
 Consider one BS which has $L_0\in\{1,2,\cdots\}$ associated users and  $K_0\in\{1,2,\cdots, K\}$ different file requests for the $K$ files stored at this BS. Note that $L_0\geq K_0$. Under the proposed multicasting transmission scheme, the BS transmits  each of the $K_0$ files  at rate $\tau $ (bits/second) and over $\frac{1}{K_0}$  of the total bandwidth $W$ (using FDMA).   All the users which request one of the $K_0$ files from this BS try to decode the file from the single multicast transmission of the file  at the BS.    In contrast, under the traditional connection-based transmission (unicast) scheme, the BS transmits one file to each of the $L_0$ users at some rate depending on the channel condition and over $\frac{1}{L_0}$  of the total bandwidth $W$ (using FDMA).  Note that in this case, $L_0-K_0$ out of $L_0$ file transmissions are redundant. Therefore, compared with the traditional connection-based transmission (unicast), this information-centric transmission (multicast) can  improve the efficiency of the utilization of the wireless medium and reduce the load of the wireless links.

\section{Performance Metric and Problem Formulation}\label{Sec:perf}

\subsection{Performance Metric}\label{Subsec:perfm}

In this paper, w.l.o.g., we study the performance of the typical user denoted as $u_0$, which is located at the origin. The index of the typical user and its serving BS is denoted as $0$. We assume all BSs are active. According to the   caching  and user association illustrated in Section~\ref{subsec:description}, when $u_{0}$ requests file $n$, the received signal of  $u_{0}$ is given by
\begin{align}
y_{0}=D_{0,0}^{-\frac{\alpha}{2}}h_{0,0}x_{0}+\sum_{\ell\in\Phi_{b}\backslash B_{n,0}}  D_{\ell,0}^{-\frac{\alpha}{2}}h_{\ell,0}x_{\ell}+n_{0},
\end{align}
where $B_{n,0}$ is the serving BS  of $u_{0}$, $D_{0,0}$ is the distance between $u_{0}$ and  $B_{n,0}$, $h_{0,0}\dis \mathcal{CN}\left(0,1\right)$ is the small-scale channel between $B_{n,0}$ and $u_{0}$, $x_{0}$ is the transmit signal from $B_{n,0}$  to $u_{0}$, $D_{\ell,0}$ is the distance between BS $\ell$ and $u_{0}$, $h_{\ell,0}\dis \mathcal{CN}\left(0,1\right)$ is the small-scale channel  between BS $\ell$ and $u_{0}$, $x_{\ell}$ is the transmit signal from BS $\ell$  to its scheduled user,  and $n_{0}\dis\mathcal{CN}\left(0,N_{0}\right)$ is the complex additive white Gaussian noise of power $N_0$. The  signal-to-interference plus noise ratio (SINR) of $u_{0}$ is given by
\begin{align}
{\rm SINR}_{n,0} = \frac{{D_{0,0}^{-\alpha}}\left|h_{0,0}\right|^{2}}{\sum_{\ell\in\Phi_{b}\backslash B_{n,0}}D_{\ell,0}^{-\alpha}\left|h_{\ell,0}\right|^{2}+\frac{N_{0}}{P}}\;.\label{eqn:SINR}
\end{align}

In addition, let $K_{n,0}\in \{1,\cdots, K\}$ denote the number of different files requested by the users associated with $B_{n,0}$, indicating the file load  of $B_{n,0}$.   Note that when $K=1$, $K_{n,0}=1$; when $K>1$,  $K_{n,0}$ is a discrete random variable, the probability mass function (p.m.f.) of which depends on the user density $\lambda_u$.   Under the proposed multicasting scheme, each of the $K_{n,0}$ requested files is sent over bandwidth $\frac{W}{K_{n,0}}$ to all the users associated with $B_{n,0}$ and  requesting this file. Thus, the corresponding channel capacity of $u_0$ is given by $C_{n,K,0}\triangleq \frac{W}{K_{n,0}}\log_{2}\left(1+{\rm SINR}_{n,0}\right)$. The dissemination of file $n$ at rate $\tau$ can be decoded correctly at $u_0$ if $C_{n,K,0}\geq \tau$.   Then, the successful transmission probability of  file $n$ requested by $u_0$ 
is given by
 \begin{align}
 q_{K,n}(\mathbf p)\triangleq{\rm Pr}\left[C_{n,K,0}\geq\tau\right]={\rm Pr}\left[\frac{W}{K_{n,0}}\log_{2}\left(1+{\rm SINR}_{n,0}\right)\geq\tau\right].\label{eqn:succ-prob-n-def}
 \end{align} 
Please note that the distributions of random variables $K_{n,0}$ and ${\rm SINR}_{n,0}$ depend on $\mathbf p$. Later, we shall see the dependence in the analysis of these distributions. Thus, we write the successful transmission probability of  file $n$ requested by $u_0$ as a function of $\mathbf p$.
Requesters are mostly concerned with whether  their desired files can be successfully received. Therefore, in this paper, we consider the successful transmission probability  of a file requested by $u_0$, also referred to as successful transmission probability, as the network performance metric.
 According to total probability theorem,  the
successful transmission probability   under the proposed multicasting scheme is given by\footnote{Note that the traditional SINR coverage probability (which does not capture resource sharing) and rate coverage probability (which only captures resource sharing among different users) cannot reflect resource sharing among different files under the multicasting scheme, and hence are not suitable to measure the performance in our case.}
\begin{align}
q_K(\mathbf p)\triangleq&\sum_{n\in \mathcal N}a_{n}q_{K,n}(\mathbf p)=\sum_{n\in \mathcal N}a_{n}{\rm Pr}\left[\frac{W}{K_{n,0}}\log_{2}\left(1+{\rm SINR}_{n,0}\right)\geq\tau\right].\label{eqn:succ-prob-def}
\end{align}
Note that we also write $q_K(\mathbf p)$ as a function of $\mathbf p$, as  $q_{K,n}(\mathbf p)$ is a function of $\mathbf p$.

Under the proposed caching and multicasting scheme  for content-oriented services in the cache-enabled  wireless network, the successful transmission probability is sufficiently different from the traditional  rate coverage probability studied for connection-oriented services \cite{WCOM13Andrews}. In particular, the successful transmission probability considered in this paper not only depends on  the physical layer parameters, such as the BS density $\lambda_b$, user density $\lambda_u$, path loss exponent $\alpha$, bandwidth $W$  and transmit signal-to-noise ratio (SNR) $\frac{P}{N_0}$, but also relies on the information-related parameters, such as the popularity distribution $\mathbf a$, the cache size $K$ and the caching distribution $\mathbf p$. While, the traditional rate coverage probability only depends on the physical layer parameters. In addition, the successful transmission probability  depends on  the physical layer parameters in a different way from the traditional rate coverage probability. 
For example, the information-centric association leads to different distributions of the locations of  serving and interfering BSs; the multicasting transmission results in different file load distributions at each BS \cite{WCOM13Andrews}; and the cache-enabled architecture makes   the content availability related to  the BS density. 
Later, in Sections~\ref{Subsec:K1-analysis} and \ref{Subsec:K-analysis}, we shall analyze the successful transmission probability under the proposed caching and multicasting scheme, i.e., \eqref{eqn:succ-prob-def},  for the unit cache size ($K=1$) and the general cache size ($K>1$), respectively.

To further illustrate the advantage of  the proposed multicasting scheme over the traditional unicasting scheme for content-oriented services, we also introduce the rate coverage probability  under the traditional unicasting  scheme for content-oriented services.\footnote{Note that in this paper, for the tractability of the analysis and obtaining first-order design insights, we assume that all the files are delivered at the same rate $\tau$. Under this assumption, there is no intricate tradeoff between unicast and multicast. We would like to consider content transmissions at multiple bit rates and multicast grouping  in future work.} Let $L_{n,0}\in \{1,2,\cdots\}$ denote the number of users associated with $B_{n,0}$. Note that $L_{n,0}\geq K_{n,0}$ is a discrete random variable, indicating the user load  of $B_{n,0}$.   Under the traditional unicasting scheme, each of the $L_{n,0}$ users is served  over bandwidth $\frac{W}{L_{n,0}}$ separately.
Similarly, the
rate coverage probability under the traditional unicasting  scheme   is given by\footnote{Note that  $q_K^{uc}(\mathbf p)$ is for content-oriented services in a cache-enabled network and is still different from the traditional  rate coverage probability for connection-oriented services\cite{WCOM13Andrews}. The purpose of introducing $q_K^{uc}(\mathbf p)$ is for demonstrating the gain of multicasting over unicasting. We will not analyze or optimize $q_K^{uc}(\mathbf p)$ in \eqref{eqn:succ-prob-def-uni}.}
\begin{align}
q_K^{uc}(\mathbf p)
=\sum_{n\in \mathcal N}a_{n}{\rm Pr}\left[\frac{W}{L_{n,0}}\log_{2}\left(1+{\rm SINR}_{n,0}\right)\geq\tau\right].\label{eqn:succ-prob-def-uni}
\end{align}
where $\tau$ can be interpreted as the rate threshold. Similarly, please note that  the distributions of  random variables $L_{n,0}$ and ${\rm SINR}_{n,0}$  depend on $\mathbf p$.  Thus, we also write $q_K^{uc}(\mathbf p)$ as a function of $\mathbf p$.
Since the file load is smaller than or equal to the user load, i.e., $K_{n,0}\leq L_{n,0}$ with probability one, we can easily conclude $q_K(\mathbf p)>q_K^{uc}(\mathbf p)$, for any given caching distribution  $\mathbf p$.
In addition, as the user density $\lambda_u$ increases, $q_K(\mathbf p)-q_K^{uc}(\mathbf p)$ increases.
This  illustrates the advantage of  the proposed multicasting scheme over the traditional unicasting scheme for content-oriented services, especially in the high user density region.
Later, we shall evaluate $q_K^{uc}(\mathbf p)$ in \eqref{eqn:succ-prob-def-uni} numerically to demonstrate  this advantage.

\subsection{Problem Formulation}

The caching and multicasting design fundamentally affects the network performance  via the caching distribution $\mathbf{p}$.  We would like to consider the optimal caching and multicasting to maximize the successful transmission probability by carefully optimizing the design parameter $\mathbf{p}$.
Specifically, we consider the following optimization problem.
\begin{Prob} [Caching and Multicasting Optimization]\label{prob:opt}
\begin{align}
\max_{\mathbf{p}} &\quad  q_K\left(\mathbf{p}\right)\nonumber\\
s.t. & \quad \eqref{eqn:cache-constr-indiv}, \eqref{eqn:cache-constr-sum},\nonumber
\end{align}
where $q_K\left(\mathbf{p}\right)$ is given by \eqref{eqn:succ-prob-def}.
\end{Prob}

Note that in this paper, we focus on the successful transmission probability  maximization  to get first-order insights into the design of cache-enabled wireless   networks.\footnote{The optimal solution to Problem~\ref{prob:opt} may result in starvation of requesters for files with low caching probabilities. We assume these starving users can be satisfied  through other service mechanisms at additional backhaul or delay costs.}  The  optimization framework in this paper can be easily applied to address  the  QoS requirements in terms of the successful transmission probability of each file, e.g., $q_{K,n}(\mathbf p)\geq Q_{K,n}$ for all $n\in \mathcal N$.
Later, in Sections~\ref{Subsec:K1-opt} and \ref{Subsec:K-opt}, we shall solve Problem~\ref{prob:opt} for the unit cache  size ($K=1$) and the general cache size ($K>1$), respectively.

\section{Analysis and Optimization for Unit Cache Size}\label{Sec:K1}
In this section, we consider the  unit cache size, i.e., $K=1$. In this case, each combination contains only one file, and there are only $I=N$ combinations. Therefore, when $K=1$, for ease of illustration, we use the file index $n\in\mathcal N$ instead of the combination index $i\in \mathcal I$, and write the caching distribution as $\mathbf p=(p_n)_{n\in \mathcal N}$. In addition, we have $K_{n,0}=1$.  In the following, we first analyze the successful transmission probability for a given design parameter $\mathbf p$. Then, we optimize the design parameter $\mathbf p$ to maximize  the successful transmission probability.

\subsection{Performance Analysis}\label{Subsec:K1-analysis}

In this part, we analyze the successful transmission probability for given  $\mathbf p$ when $K=1$. Note that different from the traditional connection-based network, in the cache-enabled wireless  network considered in this paper, there are two types of interferers, namely, i) interfering BSs storing the file requested by  $u_{0}$   (these BSs are further than the serving BS of $u_{0}$), and ii) interfering BSs storing the other files (these BSs could be closer to $u_0$  than the serving BS of $u_{0}$).
 By carefully handling these two types of interferers, we obtain $q_{1,n}\left(\mathbf{p}\right)$  using  stochastic geometry. Substituting $q_{1,n}\left(\mathbf{p}\right)$ into \eqref{eqn:succ-prob-def}, we have the following result.
\begin{Thm} [Performance for $K=1$]
The successful transmission probability $q_1\left(\mathbf{p}\right)$ of $u_{0}$ is given by
\begin{align}\label{eq:CPrate_1file_noise}
q_1\left(\mathbf{p}\right)=\sum_{n\in \mathcal N}a_{n}f_1(p_n),
\end{align}
where
\begin{align}
f_k(x)\triangleq 2\pi\lambda_{b}x\int_{0}^{\infty}&d\exp\left(-\frac{2\pi}{\alpha}x\lambda_{b}\left(2^{\frac{k\tau}{W}}-1\right)^{\frac{2}{\alpha}}d^{2}B^{'}\left(\frac{2}{\alpha},1-\frac{2}{\alpha},2^{-\frac{k\tau}{W}}\right)\right)\nonumber\\
&\times
\exp\left(-\pi\lambda_{b}xd^{2}\right)\exp\left(- \left(2^{\frac{k\tau}{W}}-1\right)d^{\alpha}\frac{N_0}{P}\right)\nonumber\\
&\times\exp\left(-\frac{2\pi}{\alpha}\left(1-x\right)\lambda_{b}\left(2^{\frac{k\tau}{W}}-1\right)^{\frac{2}{\alpha}}d^{2}B\left(\frac{2}{\alpha},1-\frac{2}{\alpha}\right)\right){\rm d}d.\label{eqn:def-f}
\end{align}
Here, $B^{'}\left(x,y,z\right)\triangleq \int_{z}^{1}u^{x-1}\left(1-u\right)^{y-1}{\rm d}u$ is the complementary incomplete Beta function, and $B(x,y)\triangleq\int_{0}^{1}u^{x-1}\left(1-u\right)^{y-1}{\rm d}u$ is the Beta function.\label{Thm:generalK1}
\end{Thm}
\begin{proof} Please refer to Appendix A.
\end{proof}

From Theorem~\ref{Thm:generalK1}, we can see that in the general region, the physical  layer parameters $\alpha$, $W$, $\lambda_b$, $\frac{P}{N_0}$ and the caching distribution $\mathbf p$ jointly affect the successful transmission probability $q_1\left(\mathbf{p}\right)$. The impacts of the physical layer parameters   and the caching distribution  on $q_1\left(\mathbf{p}\right)$  are coupled in a complex manner. 

The successful transmission probability $q_1\left(\mathbf{p}\right)$ increases with the transmit SNR $\frac{P}{N_0}$.  From  Theorem~\ref{Thm:generalK1},   we have the following corollary.

\begin{Cor} [Asymptotic Performance for $K=1$]
When $\frac{P}{N_0}\to \infty$, the successful transmission probability of $u_{0}$ is given by
\begin{align}
q_{1,\infty}\left(\mathbf{p}\right)\triangleq \lim_{\frac{P}{N_0}\to \infty}q_1\left(\mathbf{p}\right)=\sum_{n\in \mathcal N}\frac{a_np_n}{c_{2,1}+c_{1,1}p_n},
\label{eqn:CPrate_1file_nonoise}
\end{align}
where
\begin{align}
c_{1,k}\triangleq &1+\frac{2}{\alpha}\left(2^{\frac{k\tau}{W}}-1\right)^{\frac{2}{\alpha}}B^{'}\left(\frac{2}{\alpha},1-\frac{2}{\alpha},2^{-\frac{k\tau}{W}}\right)-\frac{2}{\alpha}\left(2^{\frac{k\tau}{W}}-1\right)^{\frac{2}{\alpha}}B\left(\frac{2}{\alpha},1-\frac{2}{\alpha}\right),\label{eqn:c1-def}\\
c_{2,k}\triangleq &\frac{2}{\alpha}\left(2^{\frac{k\tau}{W}}-1\right)^{\frac{2}{\alpha}} B\left(\frac{2}{\alpha},1-\frac{2}{\alpha}\right).\label{eqn:c2-def}
\end{align}
\label{Cor:generalK1}
\end{Cor}
\begin{proof} Please refer to Appendix B.
\end{proof}

From Corollary~\ref{Cor:generalK1}, we can see that in the high SNR region, the impact of the physical layer parameters $\alpha$  and $W$, captured by $c_{1,1}$ and $c_{2,1}$, and the impact of the caching distribution  $\mathbf p$ on   $q_{1,\infty}\left(\mathbf{p}\right)$ can be easily separated. Later, in Section~\ref{Subsec:K1-opt}, we shall see that this separation greatly facilitates the optimization of $q_{1,\infty}\left(\mathbf{p}\right)$ over $\mathbf p$.

Fig.~\ref{fig:verification-K1} plots the successful transmission probability  versus the transmit  SNR $\frac{P}{N_0}$.  From Fig.~\ref{fig:verification-K1}, we can see that the analytical curve for the general transmit SNR   is very close to the Monte Carlo curve. In addition, we can see that the analytical curve for the general transmit SNR asymptotically approaches the analytical curve for the asymptotic  transmit SNR. These observations verify Theorem~\ref{Thm:generalK1} and Corollary~\ref{Cor:generalK1}, respectively. From Fig.~\ref{fig:verification-K1}, we can also observe that  $q_{1,\infty}\left(\mathbf{p}\right)$ provides a simple and good approximation for $q_1\left(\mathbf{p}\right)$ in the high transmit SNR region (e.g., $\frac{P}{N_0}\geq 40$ dB). On the other hand, Fig.~\ref{fig:verification-K1} shows that   $q_1(\mathbf p)$, which does not depend on user density $\lambda_u$ (when $K=1$), is greater than  $q_1^{uc}(\mathbf p)$, which depends on the user density $\lambda_u$, as discussed in Section~\ref{Subsec:perfm}. This demonstrates the advantage of the proposed multicasting scheme over the traditional unicasting scheme.

\begin{figure}
\begin{center}
 \includegraphics[width=7cm]{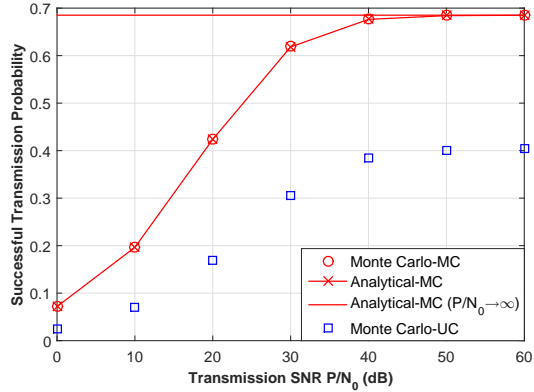}
  \end{center}
    \caption{\small{Successful transmission probability versus transmit SNR $\frac{P}{N_0}$ at $K=1$. $\lambda_b = 0.01$, $\lambda_u = 0.1$, $\alpha = 4$, $W = 10\times 10^6$, $\tau = 5\times10^5$, $N=5$, $\mathbf p = (0.6811, 0.3189,  0, 0,  0)$, and $a_n=\frac{n^{-\gamma}}{\sum_{n\in \mathcal N}n^{-\gamma}}$  with $\gamma =2$. MC stands for multicast, while UC stands for unicast. In  this paper, to simulate the large-scale network,  we use a 2-dimensional square of area  $260^2$, which is sufficiently large in our case. Note that if the simulation window size is not large enough, the observed interference would be smaller than the true interference due to the edge effect, resulting in larger successful transmission probability than the true value. In addition, the Monte Carlo results are obtained by averaging over $4\times10^6$ random realizations.}}
\label{fig:verification-K1}
\end{figure}

\subsection{Performance Optimization}\label{Subsec:K1-opt}

In this part, we study the successful transmission probability optimization in Problem \ref{prob:opt} for $K=1$, with $q_1\left(\mathbf{p}\right)$  in  \eqref{eq:CPrate_1file_noise} being the objective function.
 In general, it is difficult to ensure the convexity of $q_1\left(\mathbf{p}\right)$, as it is in a very complex form.   On the other hand, $q_1\left(\mathbf{p}\right)$ is differentiable, and  the constraints in
\eqref{eqn:cache-constr-indiv} and \eqref{eqn:cache-constr-sum}  are linear. Thus, Problem \ref{prob:opt} for $K=1$   is an optimization of a differentiable (non-convex) function over a convex set.  A local optimal solution can be obtained using  standard gradient projection methods\cite[pp. 223]{Bertsekasbooknonlinear:99}.
Here, we consider the diminishing stepsize\cite[pp. 227]{Bertsekasbooknonlinear:99} satisfying
\begin{align}
\epsilon(t)\to 0\ \text{as}\ t\to \infty,\  \sum_{t=1}^{\infty}\epsilon(t)=\infty, \sum_{t=1}^{\infty}\epsilon(t)^2<\infty, \label{eqn:stepcond}
\end{align}
and propose Algorithm \ref{alg:localK1} to obtain a local optimal solution. It is shown in \cite[pp. 229]{Bertsekasbooknonlinear:99} that $\mathbf p(t)\to \mathbf p^{\dagger}$ as $t\to \infty$, where $ \mathbf p^{\dagger}$ is a local optimal solution to Problem \ref{prob:opt} for $K=1$.
\begin{algorithm} [Local Optimal Solution for $K=1$]
\caption{Local Optimal Solution for $K=1$}
\begin{algorithmic}[1]
\STATE Initialize  $t=1$ and $p_n(1)=\frac{1}{N}$  for all $n\in \mathcal N$.
 \STATE   For all $n\in \mathcal N$, compute $\bar p_n(t+1)$ according to $\bar p_n(t+1)=p_n(t)+\epsilon(t)\frac {\partial q_1\left(\mathbf{p}(t)\right)}{\partial p_n(t)}$,
where  $\{\epsilon(t)\}$ satisfies \eqref{eqn:stepcond}.
 \STATE For all $n\in \mathcal N$, compute $p_n(t+1)$ according to $p_n(t+1)=\min\left\{\left[\bar p_n(t+1)-\nu^*\right]^+,1\right\}$,
where $\nu^*$ satisfies $\sum_{n\in \mathcal N}\min\left\{\left[\bar p_n(t+1)-\nu^*\right]^+,1\right\}=1$.
 \STATE Set $t=t+1$ and go to Step 2.
\end{algorithmic}\label{alg:localK1}
\end{algorithm}

In Step 2 of Algorithm \ref{alg:localK1}, $\frac {\partial q_1\left(\mathbf{p}(t)\right)}{\partial p_n(t)}=a_nf_1'(p_n)$, where
\begin{align}\label{eq:derivative-fk}
f_k'(x)=\frac{f_k(x)}{x}+2\pi\lambda_{b}x\int_{0}^{\infty}&d\exp\left(-\frac{2\pi}{\alpha}x\lambda_{b}\left(2^{\frac{k\tau}{W}}-1\right)^{\frac{2}{\alpha}}d^{2}B^{'}\left(\frac{2}{\alpha},1-\frac{2}{\alpha},2^{-\frac{k\tau}{W}}\right)\right)\nonumber\\
&\times
\exp\left(-\pi\lambda_{b}xd^{2}\right)\exp\left(- \left(2^{\frac{k\tau}{W}}-1\right)d^{\alpha}\frac{N_0}{P}\right)\nonumber\\
&\times\exp\left(-\frac{2\pi}{\alpha}\left(1-x\right)\lambda_{b}\left(2^{\frac{k\tau}{W}}-1\right)^{\frac{2}{\alpha}}d^{2}B\left(\frac{2}{\alpha},1-\frac{2}{\alpha}\right)\right)\nonumber\\
&\times\Bigg(-\frac{2\pi}{\alpha}\lambda_{b}\left(2^{\frac{k\tau}{W}}-1\right)^{\frac{2}{\alpha}}d^{2}B^{'}\left(\frac{2}{\alpha},1-\frac{2}{\alpha},2^{-\frac{k\tau}{W}}\right)-\pi\lambda_{b}d^{2}\nonumber\\
&+\frac{2\pi}{\alpha}\lambda_{b}\left(2^{\frac{k\tau}{W}}-1\right)^{\frac{2}{\alpha}}d^{2}B\left(\frac{2}{\alpha},1-\frac{2}{\alpha}\right)\Bigg)
{\rm d}d.
\end{align}
Note that Step 3 of Algorithm \ref{alg:localK1} is the projection of $\bar p_n(t+1)$ onto the set of the  variables satisfying the constraints in \eqref{eqn:cache-constr-indiv} and \eqref{eqn:cache-constr-sum}. In other words, $p_n(t+1)$ given in Step 3  is the solution to the following optimization problem\cite[pp. 201]{Bertsekasbooknonlinear:99}.\footnote{This can be easily shown using KKT conditions. We omit the details due to page limitation.}
\begin{align}
\min_{\mathbf p}&\quad ||\mathbf p-\bar{\mathbf p}(t+1)||^2\nonumber\\
s.t.& \quad \eqref{eqn:cache-constr-indiv},  \eqref{eqn:cache-constr-sum}.\nonumber
\end{align}
Here, $||\cdot||$ denotes the  Euclidean norm.

As computing the gradient of $q_1\left(\mathbf{p}\right)$ in each iteration involves very high complexity,  Algorithm \ref{alg:localK1} may not be suitable for  this problem when $N$ is large. 
Therefore, instead, we focus on obtaining an asymptotically (global) optimal solution when $\frac{P}{N_0}\to \infty$.  Specifically, we consider the optimization of the  asymptotic successful transmission probability, i.e., Problem \ref{prob:opt}  with  $q_{1,\infty}\left(\mathbf{p}\right)$  in \eqref{eqn:CPrate_1file_nonoise} being the objective function.
\begin{Prob} [Asymptotic  Optimization for $K=1$]\label{prob:opt-1-infty}
\begin{align}
\max_{\mathbf{p}} &\quad  q_{1,\infty}\left(\mathbf{p}\right)\nonumber\\
s.t. & \quad \eqref{eqn:cache-constr-indiv}, \eqref{eqn:cache-constr-sum}.\nonumber
\end{align}
\end{Prob}

The optimal solution $\mathbf p^{*}\triangleq (p_n^*)_{n\in \mathcal N}$ to Problem~\ref{prob:opt-1-infty} is an asymptotically optimal solution to Problem \ref{prob:opt} for $K=1$, and the resulting  successful transmission probability $q^*_{1,\infty}\triangleq q_{1,\infty}(\mathbf p^{*})$ is the asymptotically optimal successful transmission probability to Problem~\ref{prob:opt} for $K=1$. In the following, we   focus on solving Problem~\ref{prob:opt-1-infty} to obtain an asymptotically optimal solution to Problem \ref{prob:opt} for $K=1$.

It can be easily seen that Problem \ref{prob:opt-1-infty} is convex  and Slater's condition is satisfied, implying that strong duality holds. Using KKT conditions, we can solve Problem \ref{prob:opt-1-infty}.

\begin{Thm} [Asymptotically Optimal Solution for $K=1$]
An asymptotically optimal solution  $\mathbf p^*=(p_n^*)_{n\in \mathcal N}$ to Problem~\ref{prob:opt} for $K=1$  is given by
\begin{align}
p_n^*=\left[\frac{1}{c_{1,1}}\sqrt{\frac{a_nc_{2,1}}{\nu^*}}-\frac{c_{2,1}}{c_{1,1}}\right]^+, \ n\in \mathcal N,\label{eqn:opt-1-infty}
\end{align}
where $[x]^+\triangleq \max\{x,0\}$ and $\nu^*$ satisfies
\begin{align}
\sum_{n\in \mathcal N}\left[\frac{1}{c_{1,1}}\sqrt{\frac{a_nc_{2,1}}{\nu^*}}-\frac{c_{2,1}}{c_{1,1}}\right]^+=1.\label{eqn:opt-1-infty-v}
\end{align}
Here, $c_{1,1}$ and $c_{2,1}$ are given by \eqref{eqn:c1-def} and \eqref{eqn:c2-def}, respectively.
\label{Thm:solu-opt-1-infty}
\end{Thm}
\begin{proof} Please refer to Appendix C.
\end{proof}

\begin{Rem} [Interpretation of Theorem~\ref{Thm:solu-opt-1-infty}] As illustrated in Fig.~\ref{Fig:optstructure} (a), the asymptotically optimal solution $\mathbf p^*$ given by Theorem~\ref{Thm:solu-opt-1-infty}  has a reverse water-filling structure. The file popularity distribution $\mathbf a$ and the physical layer parameters (captured in  $c_{1,1}$ and $c_{2,1}$)  jointly affect   $\nu^*$.  Given  $\nu^*$, the physical layer parameters (captured in  $c_{1,1}$ and $c_{2,1}$) affect the  caching probabilities of all the files  in the same way, while the popularity of file $n$  (i.e., $a_n$) only affects the caching probability of file $n$  (i.e., $p^*_n$).
From Theorem~\ref{Thm:solu-opt-1-infty}, we have  $p_1^*\geq p_2^*\geq\cdots\geq p_N^*$, as $a_{1}\ge a_{2}\ldots\ge a_{N}$. In other words, files of higher popularity get more storage resources. In addition, there may exist $\bar n$ ($1<\bar n<N$) such that $p_n^*>0$ for all  $n<\bar n$, and $p_n^*=0$ for all  $n\geq\bar n$. In other words, some files of lower popularity may not be stored.   For a popularity distribution with a heavy tail, $\bar n$ is large, i.e.,   more  different files  can be stored.
\end{Rem}

\begin{figure}
\begin{center}
  \subfigure[\small{$\mathbf p^*$ for $K=1$.}]
  {\resizebox{7cm}{!}{\includegraphics{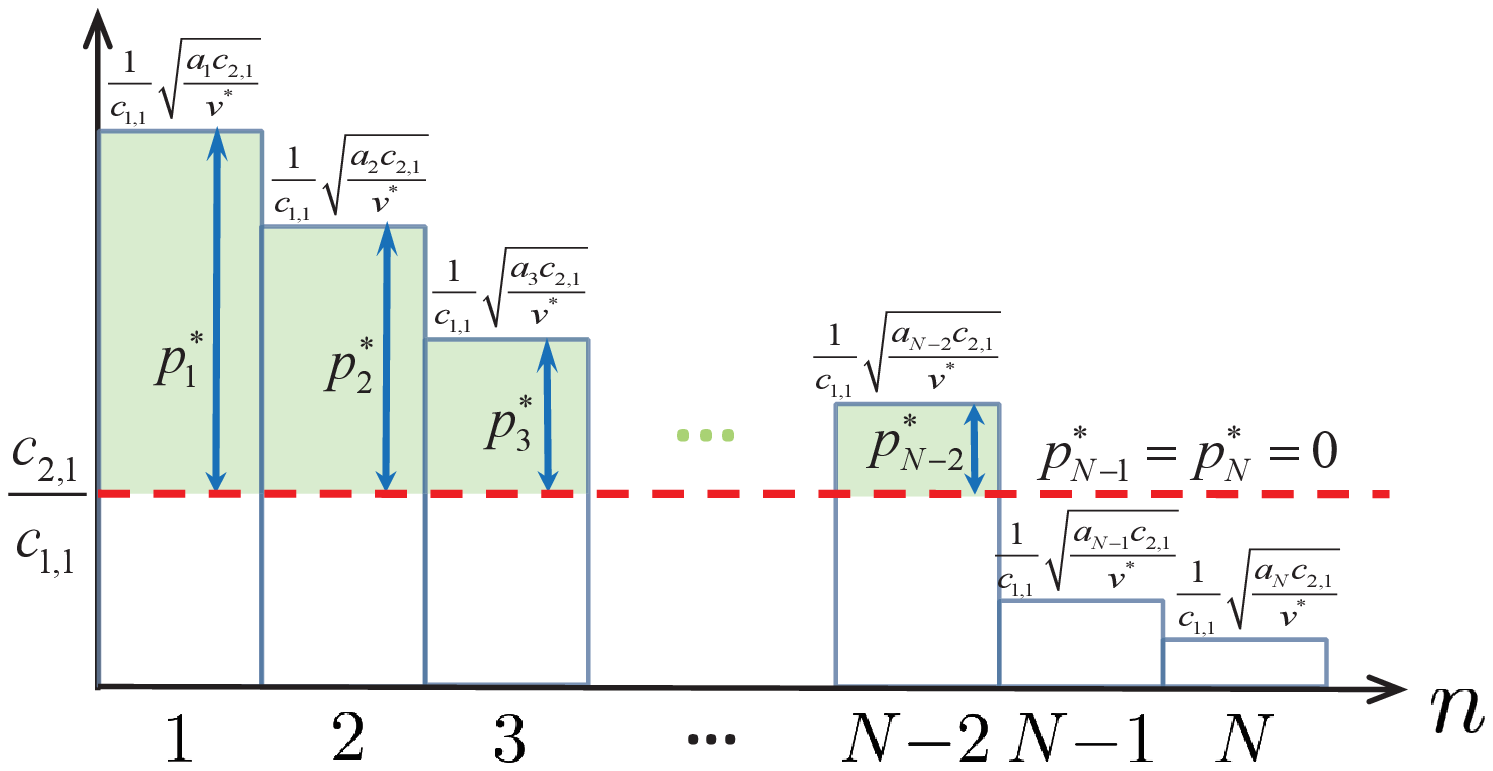}}}
\quad
  \subfigure[\small{$\mathbf T^*$ for $K>1$.}]
  {\resizebox{7cm}{!}{\includegraphics{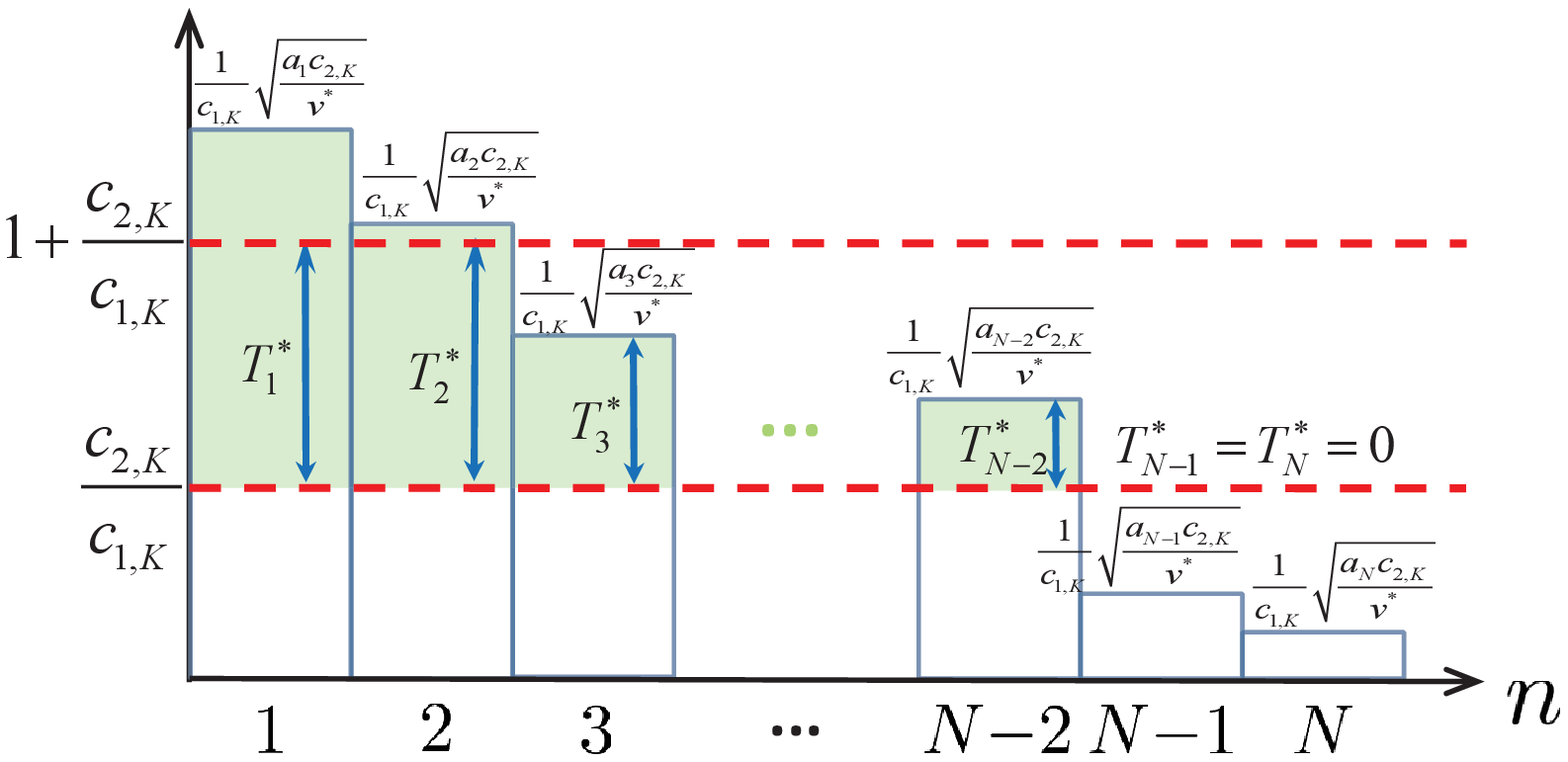}}}
  \end{center}
    \caption{Illustration of the optimality structure.}
\label{Fig:optstructure}
\end{figure}


As the water-level in the traditional water-filling power control,   the root $\nu^*$ to the equation in \eqref{eqn:opt-1-infty-v} can be easily solved. Thus, we can efficiently compute   $\mathbf p^*$ based on Theorem~\ref{Thm:solu-opt-1-infty}. In addition,  from Theorem~\ref{Thm:solu-opt-1-infty}, we have the following corollary.

 \begin{Cor}  [Asymptotically Optimal Solution for $K=1$]
If
$\frac{\sqrt{a_{N}}}{\sum_{n\in \mathcal N}\sqrt{a_{n}}}>\frac{c_{2,1}}{c_{1,1}+c_{2,1}N}$, then we have

\begin{align}
&p_{n}^{*}=\left(1+\frac{c_{2,1}}{c_{1,1}}N\right)\frac{\sqrt{a_{n}}}{\sum_{n\in \mathcal N}\sqrt{a_{n}}}-\frac{c_{2,1}}{c_{1,1}}>0,\ n \in \mathcal N,\\
&q^*_{1,\infty}=\frac{1}{c_{1,1}}\left(1-\frac{(\sum_{n\in \mathcal N}\sqrt{a_{n}})^2}{N+\frac{c_{1,1}}{c_{2,1}}}\right).
\end{align}
Here, $c_{1,1}$ and $c_{2,1}$ are given by \eqref{eqn:c1-def} and \eqref{eqn:c2-def}, respectively.
\label{Cor:solu-opt-1-infty}
\end{Cor}

Note that the condition  $\frac{\sqrt{a_{N}}}{\sum_{n\in \mathcal N}\sqrt{a_{n}}}>\frac{c_{2,1}}{c_{1,1}+c_{2,1}N}$  means that the popularity distribution has a relatively heavy tail.
Under this condition, all files are stored, and interestingly, the caching probability of  file $n$ increases  linearly with   $\frac{\sqrt{a_{n}}}{\sum_{n\in \mathcal N}\sqrt{a_{n}}}$  instead of $a_n$.

\section{Analysis and Optimization for General Cache Size}\label{Sec:K}
In this section, we consider the general cache size, i.e., $K>1$. We first analyze the successful transmission probability for a given design parameter $\mathbf p$, by adopting appropriate approximations. Then, we study   the  maximization of  the successful transmission probability over $\mathbf p$.

\subsection{Performance Analysis}\label{Subsec:K-analysis}

In this part, we would like to analyze the successful transmission probability $q_K(\mathbf p)$ in  \eqref{eqn:succ-prob-def} for a given design parameter $\mathbf p$ when $K>1$.  In general, for all $n\in\mathcal N$,  file load $K_{n,0}$ and SINR ${\rm SINR}_{n,0}$ are correlated, as BSs with larger association regions have higher file load and lower SINR (due to larger user to BS distance) \cite{AndrewsTWCOffloading14}. However, the exact relationship between $K_{n,0}$ and ${\rm SINR}_{n,0}$ is very complex and is still not known. For the tractability of the analysis, as in \cite{AndrewsTWCOffloading14} and \cite{WCOM13Andrews}, the dependence is ignored.
Then, from \eqref{eqn:succ-prob-def}, we have
\begin{align}
q_K(\mathbf p)
=&\sum_{n\in \mathcal N}a_{n}\sum_{k=1}^K\Pr \left[K_{n,0}=k\right]{\rm Pr}\left[\frac{W}{k}\log_{2}\left(1+{\rm SINR}_{n,0}\right)\geq\tau \right].
\end{align}
First, we calculate the p.m.f. of random file load $K_{n,0}$.  In calculating  $\Pr \left[K_{n,0}=k\right]$, we need
the probability density function (p.d.f.) of the size of the Voronoi cell of $B_{n,0}$ w.r.t. file $m\in \mathcal N_{i,-n}\triangleq  \mathcal N_i \setminus \{n\}$  when  $B_{n,0}$ contains combination $i\in \mathcal I_n$. However, this p.d.f. is very complex and is still unknown. For the tractability of  the analysis,  we approximate this p.d.f. based on a tractable approximated form of the  p.d.f. of the size of the Voronoi cell to which a randomly chosen user belongs\cite{SGcellsize13}, which is widely used in existing literature\cite{AndrewsTWCOffloading14,WCOM13Andrews}. Under this approximation, we calculate  the p.m.f. of $K_{n,0}$. The accuracy of the approximation will be demonstrated in Fig.~\ref{fig:verification-Kmulti} and Table~\ref{tab:Error}.
\begin{Lem} [p.m.f. of $K_{n,0}$] The p.m.f. of $K_{n,0}$ is given by
\begin{align}
\Pr \left[K_{n,0}=k\right]\approx&\sum_{i\in \mathcal I_n}\frac{p_i}{T_n}\sum_{\mathcal N_i^1\in \mathcal{SN}_i^1(k-1) }G_{n,i}(\mathcal N_i^1,\mathbf T_{i,-n}), \ k=1,\cdots, K,\label{eqn:K-pmf}
\end{align}
where
$\mathcal{SN}_i^1(k-1)\triangleq \left\{\mathcal N_i^1 \subseteq \mathcal N_{i,-n} :|\mathcal N_i^1|=k-1\right\}$,
$\mathbf T_{i,-n}\triangleq\left(T_m\right)_{m\in \mathcal N_{i,-n}}$,
$W_m(T_m)\triangleq1+3.5^{-1}\frac{a_m\lambda_u}{T_m\lambda_b}$ and $G_{n,i}(\mathcal N_i^1,\mathbf T_{i,-n})
\triangleq \prod_{m\in \mathcal N_i^1}\left(1-W_m(T_m)^{-4.5}\right)\prod_{m\in {\mathcal N_{i,-n}\setminus \mathcal N_i^1}}W_m(T_m)^{-4.5}.$
\label{Lem:pmf-K}
\end{Lem}

\begin{proof}
Please refer to Appendix D. \end{proof}

From Lemma~\ref{Lem:pmf-K}, we have the following corollary.
\begin{Cor} [p.m.f. of $K_{n,0}$ when $\lambda_u\to \infty$] When $\lambda_u\to \infty$, 
we have
\begin{align}
\lim_{\lambda_u\to \infty}\Pr \left[K_{n,0}=k\right]=
\begin{cases}
0, & k=1,\cdots, K-1\\
1, & k=K
\end{cases}.\label{eqn:K-pmf-aymp}
\end{align}\label{Cor:pmf-K-asymp}
\end{Cor}

By Corollary~\ref{Cor:pmf-K-asymp}, we know  that the (approximated) random variable $K_{n,0}\in\{1,\cdots, K\}$ coverages to the constant $K$ in distribution, as $\lambda_u\to \infty$. Note that in the high user density region, the exact file load for each BS turns to $K$. Thus, Corollary~\ref{Cor:pmf-K-asymp}  indicates  that the approximation error  used for obtaining the p.m.f. of $K_{n,0}$ in Lemma~\ref{Lem:pmf-K}  is asymptotically negligible.


Next, we calculate $\Pr\left[\frac{W}{k}\log_{2}\left(1+{\rm SINR}_{n,0}\right)\geq\tau\right]$. Similar to the case for $K=1$, there are two types of interferers, namely, i) interfering BSs storing the combinations containing the file requested by $u_{0}$, and ii) interfering BSs without  the desired file of $u_0$. By carefully handling these two types of interferers, we obtain $\Pr\left[\frac{W}{k}\log_{2}\left(1+{\rm SINR}_{n,0}\right)\geq\tau\right]=f_k(T_n)$  using stochastic geometry. Then, based on Lemma~\ref{Lem:pmf-K} and $\Pr\left[\frac{W}{k}\log_{2}\left(1+{\rm SINR}_{n,0}\right)\geq\tau\right]=f_k(T_n)$, we can obtain $q_K\left(\mathbf{p}\right)$ for $K>1$.
\begin{Thm} [Performance for $K>1$]
The successful transmission probability $q_K\left(\mathbf{p}\right)$ of $u_{0}$ is given by
\begin{align}\label{eq:CPrate_multifile_noise}
q_K\left(\mathbf{p}\right)=\sum_{n\in \mathcal N}a_{n} \sum_{k=1}^K \Pr [K_{n,0}=k]f_k(T_n),
\end{align}
where $\Pr [K_{n,0}=k]$ is given by Lemma~\ref{Lem:pmf-K}, $f_k(T_n)$ is given by \eqref{eqn:def-f} and $T_n$ is given by \eqref{eqn:def-T-n}.\label{Thm:generalKmulti}
\end{Thm}
\begin{proof} Please refer to Appendix E.
\end{proof}

Similarly, from Theorem~\ref{Thm:generalKmulti}, we can see that in the general region, the physical  layer parameters $\alpha$, $W$, $\lambda_b$, $\lambda_u$ and $\frac{P}{N_0}$ and the caching distribution $\mathbf p$ jointly affect the successful transmission probability $q_K\left(\mathbf{p}\right)$, by affecting $\Pr [K_{n,0}=k]$ and $f_k(T_n)$ for all $k=1,\cdots, K$.  The impacts of the physical layer parameters   and  the caching  distribution on $q_K\left(\mathbf{p}\right)$  are coupled in a complex manner.

The successful transmission probability
$q_K\left(\mathbf{p}\right)$ increases with 
$\frac{P}{N_0}$  and   decreases with  $\lambda_u$. From  Theorem~\ref{Thm:generalKmulti} and  Corollary~\ref{Cor:pmf-K-asymp},
we have the following corollary.

\begin{Cor}[Asymptotic Performance for $K>1$]
When $\frac{P}{N_{0}}\to\infty$  and $\lambda_u\to \infty$, the successful transmission probability  of $u_{0}$ is given by
\begin{align}\label{eq:CPrate_multifile_nonoise}
q_{K,\infty}\left(\mathbf{p}\right)\triangleq \lim_{\lambda_u\to \infty,\frac{P}{N_{0}}\to\infty}q_K\left(\mathbf{p}\right)=\sum_{n\in \mathcal N}\frac{a_{n}T_n}{c_{2,K}+c_{1,K}T_n}.
\end{align}
Here, $c_{1,K}$ and $c_{2,K}$ are given by \eqref{eqn:c1-def} and \eqref{eqn:c2-def}, respectively,  and $T_n$ is given by \eqref{eqn:def-T-n}.
\label{Cor:generalKmulti}
\end{Cor}
\begin{proof}
Corollary~\ref{Cor:generalKmulti} can be proved in a similar way to Corollary~\ref{Cor:generalK1}. We omit the details due to page limitation.
\end{proof}

\begin{figure}
\begin{center}
  \subfigure[\small{Successful transmission probability versus $\frac{P}{N_0}$ at  $\lambda_u = 0.1$.}]
  {\resizebox{7cm}{!}{\includegraphics{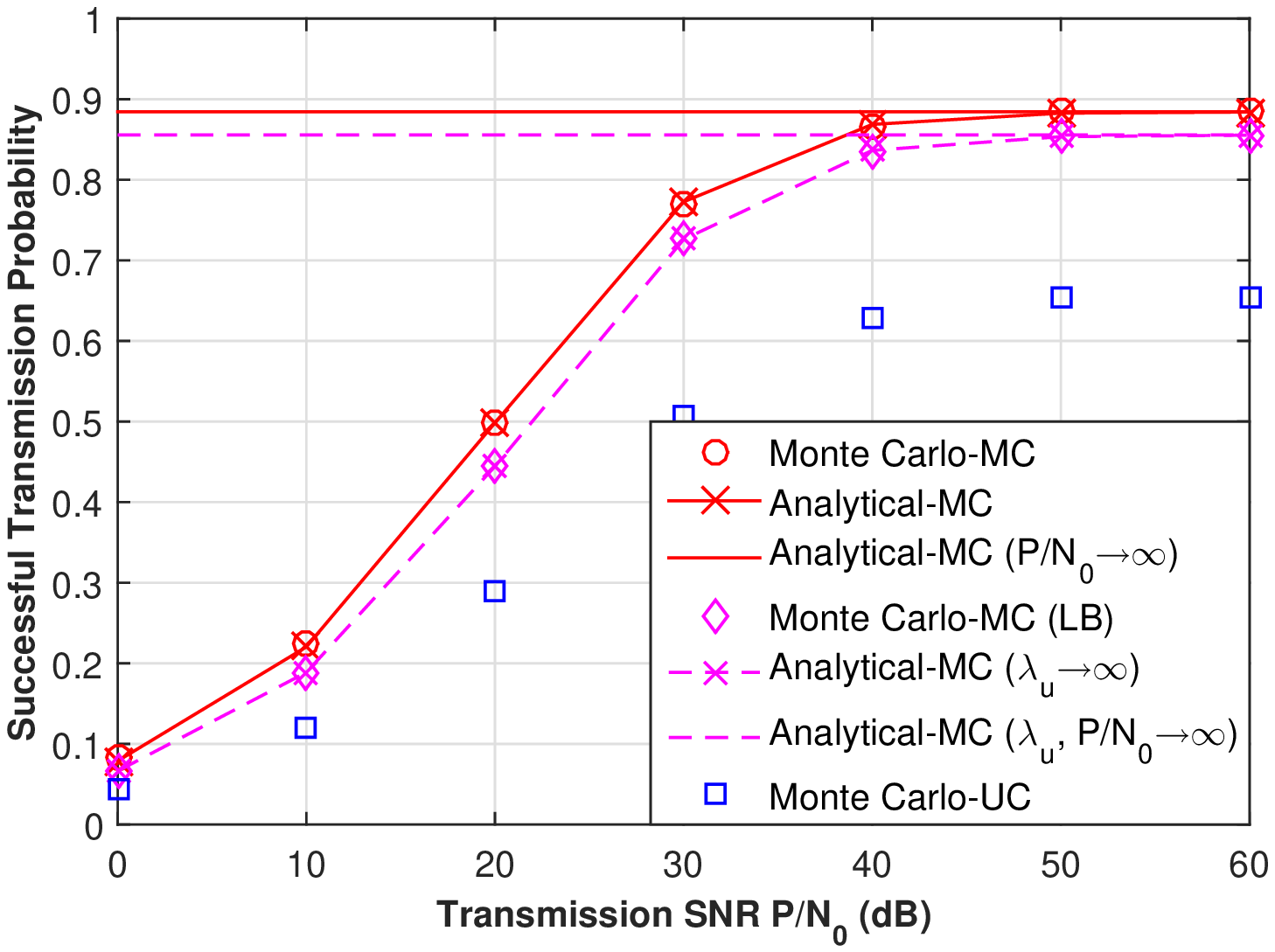}}}\quad
  \subfigure[\small{Successful transmission probability versus  $\lambda_u$ at $\frac{P}{N_0}=30$ dB.}]
  {\resizebox{7cm}{!}{\includegraphics{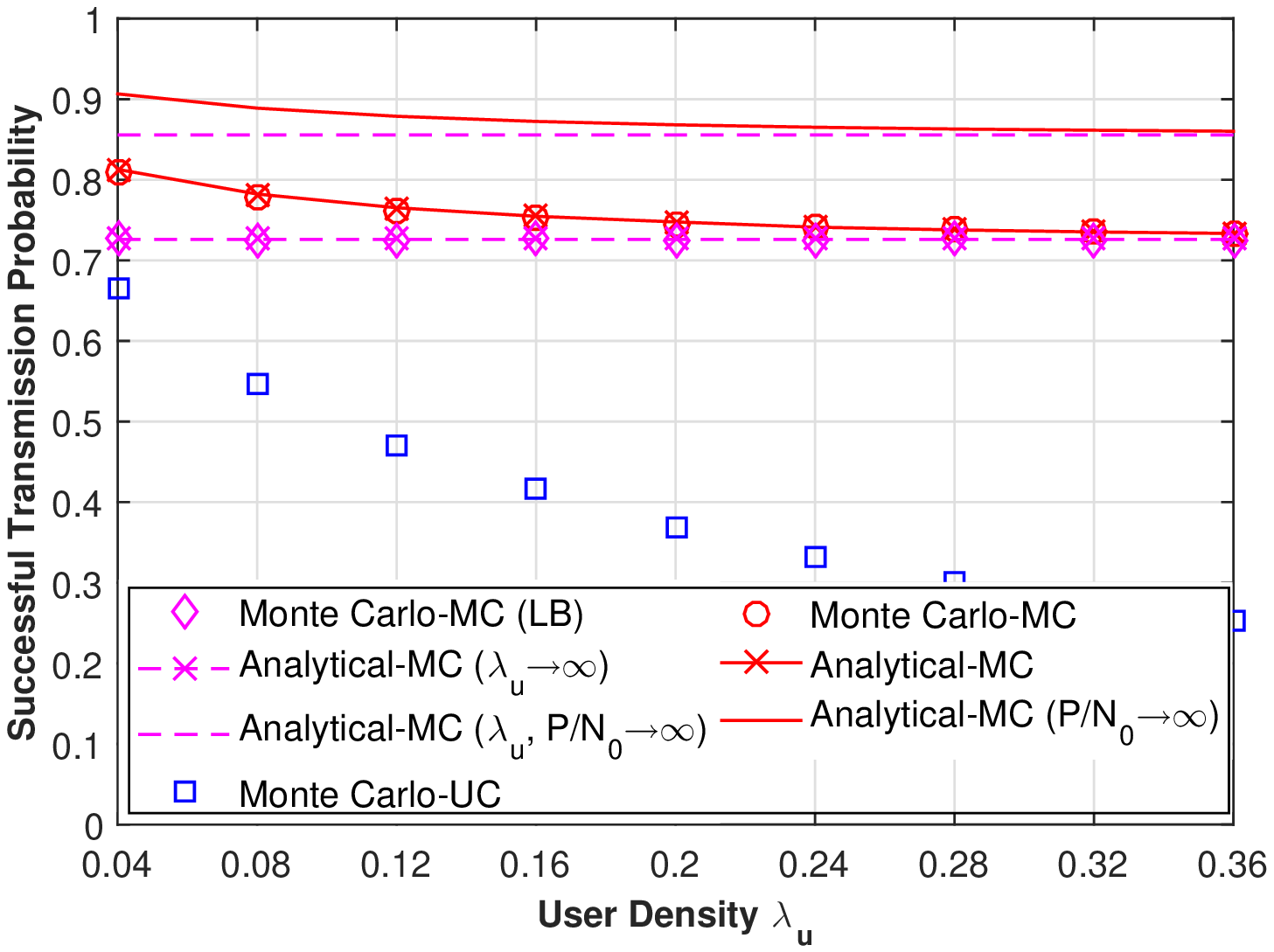}}}
  \end{center}
    \caption{\small{Successful transmission probability versus transmit SNR $\frac{P}{N_0}$  and user density $\lambda_u$ at $K=4$.  $\lambda_b = 0.01$, $\alpha = 4$, $W = 10\times 10^6$, $\tau = 5\times10^5$, $N=5$, $\mathbf p = (0.6811, 0.3189,  0, 0,  0)$, and $a_n=\frac{n^{-\gamma}}{\sum_{n\in \mathcal N}n^{-\gamma}}$  with $\gamma =2$. Here, $\mathcal N_1=\{1,2,3,4\}$, $\mathcal N_2=\{1,2,3,5\}$, $\mathcal N_3=\{1,2,4,5\}$, $\mathcal N_4=\{1,3,4,5\}$, and $\mathcal N_5=\{2,3,4,5\}$.}}
\label{fig:verification-Kmulti}
\end{figure}

Similarly, from Corollary~\ref{Cor:generalKmulti}, we can see that in the high SNR and user density region, the impact of the physical layer parameters $\alpha$  and $W$, captured by $c_{1,K}$ and $c_{2,K}$, and the impact of the caching distribution   $\mathbf p$ on the successful transmission probability $q_{K,\infty}\left(\mathbf{p}\right)$ can be easily separated. Later, in Section~\ref{Subsec:K-opt}, we shall see that this separation greatly facilitates the optimization of $q_{K,\infty}\left(\mathbf{p}\right)$.

Fig.~\ref{fig:verification-Kmulti} plots the successful transmission probability for $K>1$ versus the  transmit  SNR $\frac{P}{N_0}$ and the user density $\lambda_u$.
 Fig.~\ref{fig:verification-Kmulti}  verifies Theorem~\ref{Thm:generalKmulti} and Corollary~\ref{Cor:generalKmulti}, and demonstrates the accuracy of the approximations adopted.
 Fig.~\ref{fig:verification-Kmulti} also indicates that $q_{K,\infty}\left(\mathbf{p}\right)$  provides a simple and good approximation for  $q_K\left(\mathbf{p}\right)$  in  the high transmit SNR region (e.g., $\frac{P}{N_0}\geq 40$ dB) and the high user density region (e.g., $\lambda_u\geq 0.16$) .
On the other hand, Fig.~\ref{fig:verification-Kmulti} shows that  $q_K(\mathbf p)$  is greater than  $q_K^{uc}(\mathbf p)$; $q_K(\mathbf p)$ and  $q_K^{uc}(\mathbf p)$ both decrease with  $\lambda_u$; and  the gap   $q_K(\mathbf p)-q_K^{uc}(\mathbf p)$ increases with $\lambda_u$. These observations verify the discussions  in Section~\ref{Subsec:perfm} and demonstrate the advantage of the proposed multicasting scheme over the traditional unicasting scheme.
Table~\ref{tab:Error} further illustrates the approximation error 
versus the number of files $N$. From Table~\ref{tab:Error},  we can see that under the parameters considered in the simulation, the approximation error is very small, and there is no observable increase of the approximation error with $N$.

\begin{table}[h]
 \centering
 \begin{tabular}{|c|c|c|c|c|c|}
 \hline
 $N$ & 200 & 400 & 600 & 800 & 1000 \\
 \hline
Monte Carlo-MC & 0.5051 & 0.4822 & 0.4705 & 0.4636 & 0.4582 \\
\hline
Analytical-MC & 0.5035 & 0.4803 & 0.4691 & 0.4620 & 0.4568 \\
\hline
Approximation Error & 0.0016 & 0.0019 & 0.0014 & 0.0016 & 0.0014 \\
\hline
Relative Approximation Error & 0.32\% & 0.39\% & 0.30\% & 0.35\% & 0.31\% \\
 \hline
 \end{tabular}
 \caption{\small{Approximation error between the analytical successful transmission probability in Theorem~\ref{Thm:generalKmulti} and the numerical successful transmission probability. $W=10\times10^6$, $\tau=5\times10^5$, $\lambda_b=0.01$, $\lambda_u=0.1$, $K=20$, $\alpha=4$, $\frac{P}{N_0}=30$ dB,  $a_n=\frac{n^{-\gamma}}{\sum_{n\in \mathcal N}n^{-\gamma}}$  with $\gamma =1.2$, and $\mathbf p= \mathbf p^*$ given by Problem~\ref{Prob:asm-improvement}.}}\label{tab:Error}
\end{table}

\subsection{Performance Optimization}\label{Subsec:K-opt}
In this part, we consider the successful transmission probability optimization in Problem~\ref{prob:opt} for $K>1$, with $q_K\left(\mathbf{p}\right)$  in  \eqref{eq:CPrate_multifile_noise} being the objective function.
Similarly, Problem \ref{prob:opt} for $K>1$   is an optimization of a differentiable (non-convex) function over a convex set.  A local optimal solution can be obtained using the standard gradient projection method in  Algorithm \ref{alg:localK}. It is shown in \cite[pp. 229]{Bertsekasbooknonlinear:99} that $\mathbf p(t)\to \mathbf p^{\dagger}$ as $t\to \infty$, where $ \mathbf p^{\dagger}$ is a local optimal solution to Problem \ref{prob:opt} for $K>1$.

\begin{algorithm}[t]
\caption{Local Optimal Solution for $K>1$}
\begin{algorithmic}[1]
\STATE Initialize  $t=1$ and $p_i(1)=\frac{1}{I}$  for all $i\in \mathcal I$.
 \STATE   For all $i\in \mathcal I$, compute $\bar p_i(t+1)$ according to $\bar p_i(t+1)=p_i(t)+\epsilon(t)\frac {\partial q_K\left(\mathbf{p}(t)\right)}{\partial p_i(t)}$,
where  $\{\epsilon(t)\}$ satisfies \eqref{eqn:stepcond}.
 \STATE For all $i\in \mathcal I$, compute $p_i(t+1)$ according to $p_i(t+1)=\min\left\{\left[\bar p_i(t+1)-\nu^*\right]^+,1\right\}$,
where $\nu^*$ satisfies $\sum_{i\in \mathcal I}\min\left\{\left[\bar p_i(t+1)-\nu^*\right]^+,1\right\}=1$.
 \STATE Set $t=t+1$ and go to Step 2. 
\end{algorithmic}\label{alg:localK}
\end{algorithm}
In Step 2 of Algorithm~\ref{alg:localK}, $\frac {\partial q_K\left(\mathbf{p}(t)\right)}{\partial p_i(t)}$ is given by
\begin{align}
\frac{\partial q_K(\mathbf p)}{\partial p_{i}}=&\sum_{n\in \mathcal N_{i}}a_n\sum\limits_{k=1}^K\frac{f_k(T_n)}{T_n}\sum_{\mathcal N_i^1 \in \mathcal {SN}_i^1(k-1)}G_{n,i}(\mathcal N_i^1,\mathbf T_{i,-n})\nonumber\\
&+\sum_{m'\in \mathcal N}\Bigg(\sum_{n\in \mathcal N, n\ne m' }a_n\sum_{k=1}^K\frac{f_k(T_n)}{T_n}\sum_{i'\in {\mathcal I}_n, i'\in {\mathcal I}_{m'} }p_{i'}\frac{4.5a_{m'}\lambda_u}{3.5T_{m'}^2\lambda_b}W_{m'}(T_{m'})^{-1}\nonumber\\
&\times\bigg(
\sum_{\substack {\mathcal N_{i'}^1 \in \mathcal {SN}_{i'}^1(k-1) \\ m' \notin \mathcal N_{i'}^1}}G_{n,i'}(\mathcal N_{i'}^1,\mathbf T_{i',-n})
-\frac{W_{m'}(T_{m'})^{-4.5}}{1-W_{m'}(T_{m'})^{-4.5}}
\sum_{\substack {\mathcal N_{i'}^1 \in \mathcal {SN}_{i'}^1(k-1) \\ m' \in \mathcal N_{i'}^1}}G_{n,i'}(\mathcal N_{i'}^1,\mathbf T_{i',-n})\bigg)\nonumber\\
&+a_{m'}\sum_{k=1}^K\frac{f^{'}_k(T_{m'})T_{m'}-f_k(T_{m'})}{T_{m'}^2}\sum_{i'\in I_{m'}}p_{i'}\sum_{\mathcal N_{i'}^1 \in \mathcal {SN}_{i'}^1(k-1)}G_{m',i'}(\mathcal N_{i'}^1,\mathbf T_{i',-m'})\Bigg) \mathbf 1[i\in I_{m'}],\label{eqn:derivative-K}
\end{align}
where $\mathbf 1[\cdot]$ denotes the indicator function and $f_k'(x)$ is given by \eqref{eq:derivative-fk}.
Step 3 is the projection of $\bar p_i(t+1)$ onto the set of the  variables satisfying the constraints in \eqref{eqn:cache-constr-indiv} and \eqref{eqn:cache-constr-sum}.

As there are $I= \binom{N}{K}$ optimization variables $\mathbf p=(p_i)_{i\in \mathcal I}$ in  Problem \ref{prob:opt} for $K>1$ and  calculating $\frac {\partial q_K\left(\mathbf{p}(t)\right)}{\partial p_i(t)}$ for each $i\in\mathcal I$ in each iteration involves high complexity (due to the complex form of $q_K\left(\mathbf{p}\right)$  in  \eqref{eq:CPrate_multifile_noise}),
 the computational complexity of a local optimal solution using Algorithm~\ref{alg:localK}  is very high for large $N$.  In the following, we propose a two-step optimization framework to obtain an asymptotically optimal solution with manageable complexity and superior performance. Specifically, we first identify a set of asymptotically optimal solutions in the high SNR and user density region, by optimizing  $q_{K,\infty}\left(\mathbf{p}\right)$ in \eqref{eq:CPrate_multifile_nonoise} among all feasible solutions. Then, we obtain the asymptotically optimal solution  achieving the optimal successful transmission probability  in the general region among the set of asymptotically optimal solutions, referred to as the best asymptotically optimal solution, by optimizing   $q_K\left(\mathbf{p}\right)$ in \eqref{eq:CPrate_multifile_noise} within this set. Later, in Fig.~\ref{fig:comparison_local_linear}, we shall verify that the best asymptotically optimal solution obtained by the two-step optimization framework is outstanding, in both performance and complexity, by comparing with the local optimal solution.

First, we identify a set of asymptotically optimal solutions in the high SNR and user density region. Specifically, we consider the optimization of the asymptotic successful transmission probability $q_{K,\infty}\left(\mathbf{p}\right)$, i.e., Problem \ref{prob:opt} with
$q_{K,\infty}\left(\mathbf{p}\right)$ in
\eqref{eq:CPrate_multifile_nonoise}  being the objective function, which has a much simpler form than  $q_K\left(\mathbf{p}\right)$  in  \eqref{eq:CPrate_multifile_noise}.

\begin{Prob} [Asymptotic Optimization for $K>1$]\label{prob:opt-K-infty-p}
\begin{align}
\max_{\mathbf{p}} &\quad
q_{K,\infty}\left(\mathbf{p}\right) \nonumber\\
s.t. & \quad \eqref{eqn:cache-constr-indiv}, \eqref{eqn:cache-constr-sum}.\nonumber
\end{align}
\end{Prob}

An optimal solution $\mathbf p^{*}\triangleq (p_i^*)_{i\in \mathcal I}$ to Problem~\ref{prob:opt-K-infty-p} is an asymptotically optimal solution to Problem~\ref{prob:opt} for $K>1$, and the resulting successful transmission probability $q^*_{K,\infty}\triangleq q_{K,\infty}(\mathbf p^{*})$ is the asymptotically optimal successful transmission probability to  Problem~\ref{prob:opt} for $K>1$. In the following, we  focus on solving Problem~\ref{prob:opt-K-infty-p} to obtain an asymptotically optimal solution to Problem~\ref{prob:opt}.

Note that $q_K\left(\mathbf{p}\right)$  in  \eqref{eq:CPrate_multifile_noise} is a function of $\mathbf T\triangleq (T_n)_{n\in \mathcal N}$. Thus, we also write $q_{K, \infty}\left(\mathbf{p}\right)$  as $q_{K,\infty}\left(\mathbf{T}\right)$.
We would like to further simplify Problem~\ref{prob:opt-K-infty-p}   by exploring the relationship between $\mathbf T$ and $\mathbf p$. Now, we introduce a new optimization problem w.r.t. $\mathbf T$.
\begin{Prob} [Equivalent Asymptotic Optimization for $K>1$]\label{prob:opt-K-infty}
\begin{align}
\max_{\mathbf T} &\quad  q_{K,\infty}\left(\mathbf{T}\right)\nonumber\\
s.t. & \quad 0\leq T_n\leq 1,\ n\in \mathcal N, \label{eqn:cache-constr-indiv-t}\\
& \quad \sum_{n\in \mathcal N}T_n=K. \label{eqn:cache-constr-sum-t}
\end{align}
\end{Prob}
The optimal solution to Problem~\ref{prob:opt-K-infty} is denoted as $\mathbf T^*\triangleq (T_n^*)_{n\in \mathcal N}$, and the optimal value of Problem~\ref{prob:opt-K-infty} is given by $ q_{K,\infty}\left(\mathbf{T}^*\right)$. The following lemma shows that  Problem~\ref{prob:opt-K-infty-p}   and Problem~\ref{prob:opt-K-infty} are equivalent.

\begin{Lem}[Equivalence between Problem~\ref{prob:opt-K-infty-p}  and Problem~\ref{prob:opt-K-infty}] For any feasible  solution $\mathbf T$ to Problem~\ref{prob:opt-K-infty}, there exists  $\mathbf p$ satisfying \eqref{eqn:def-T-n} which is  a feasible solution to Problem~\ref{prob:opt-K-infty-p}. For any feasible solution $\mathbf p$ to Problem~\ref{prob:opt-K-infty-p}, $\mathbf T$ satisfying \eqref{eqn:def-T-n} is a feasible solution to Problem~\ref{prob:opt-K-infty}. Furthermore,
the optimal values of Problem~\ref{prob:opt-K-infty-p}  and Problem~\ref{prob:opt-K-infty} are the same, i.e., $q_{K,\infty}\left(\mathbf{p}^*\right)=q_{K,\infty}\left(\mathbf{T}^*\right)=q_{K,\infty}^*$.\label{Lem:equivalence}
\end{Lem}
\begin{proof} Please refer to Appendix F.
\end{proof}

Based on Lemma \ref{Lem:equivalence}, we first focus on  solving Problem \ref{prob:opt-K-infty} for $\mathbf T^*$.
It can be easily seen that Problem \ref{prob:opt-K-infty} is convex  and Slater's condition is satisfied, implying strong duality holds.
Using KKT conditions, we  can solve Problem~\ref{prob:opt-K-infty}.

\begin{Thm} [Optimization of $\mathbf T$]
The  optimal solution  $\mathbf T^*=(T_n^*)_{n\in \mathcal N}$ to Problem~\ref{prob:opt-K-infty}  is given by
 \begin{align}
T_n^*=\min\left\{\left[\frac{1}{c_{1,K}}\sqrt{\frac{a_nc_{2,K}}{\nu^*}}-\frac{c_{2,K}}{c_{1,K}}\right]^+,1\right\} , \ n\in \mathcal N,\label{eqn:opt-K-infty}
\end{align}
where  $\nu^*$ satisfies
\begin{align}
\sum_{n\in \mathcal N}\min\left\{\left[\frac{1}{c_{1,K}}\sqrt{\frac{a_nc_{2,K}}{\nu^*}}-\frac{c_{2,K}}{c_{1,K}}\right]^+,1\right\}=K.\label{eqn:opt-K-infty-v}
\end{align}
Here, $c_{1,K}$ and $c_{2,K}$ are given by \eqref{eqn:c1-def} and \eqref{eqn:c2-def}, respectively.
\label{Thm:solu-opt-K-infty-t}
\end{Thm}
\begin{proof} 
Theorem~\ref{Thm:solu-opt-K-infty-t} can be proved in a similar way to Theorem~\ref{Thm:solu-opt-1-infty}. We omit the details due to page limitation.
\end{proof}

\begin{Rem} [Interpretation of Theorem~\ref{Thm:solu-opt-K-infty-t}] As illustrated in Fig.~\ref{Fig:optstructure} (b), the  structure of the optimal solution $\mathbf T^*$  given  by Theorem~\ref{Thm:solu-opt-K-infty-t}  is similar to the reverse water-filling structure. The file popularity distribution $\mathbf a$ and the physical layer parameters (captured in  $c_{1,K}$ and $c_{2,K}$) jointly affect   $\nu^*$.  Given $\nu^*$, the physical layer parameters (captured in  $c_{1,K}$ and $c_{2,K}$) affect the  caching probabilities of all the files  in the same way, while the popularity of file $n$  (i.e., $a_n$) only affects the caching probability of file $n$  (i.e., $T^*_n$). These features for $\mathbf T^*$  are similar to those for  $\mathbf p^*$ given by Theorem~\ref{Thm:solu-opt-1-infty}.
Similar to Theorem~\ref{Thm:solu-opt-1-infty},  from Theorem~\ref{Thm:solu-opt-K-infty-t}, we know that files of higher popularity get more storage resources, and  some files of lower popularity may not be stored.   In addition, for a popularity distribution with a heavy tail,    more  different files  can be stored. The optimal structure is similar to the reverse water-filling structure. \label{Rem:K}
\end{Rem}

From Theorem~\ref{Thm:solu-opt-K-infty-t}, we have the following corollary.

\begin{Cor}  [Asymptotically Optimal Solution for $K>1$]
If  $\frac{\sqrt{a_{1}}}{\sum_{n\in \mathcal N}\sqrt{a_{n}}}\leq\frac{c_{1,K}+c_{2,K}}{c_{1,K}K+c_{2,K}N}$ and $\frac{\sqrt{a_{n}}}{\sum_{n\in \mathcal N}\sqrt{a_{n}}}>\frac{c_{2,K}}{c_{1,K}K+c_{2,K}N}$, then we have
\begin{align}
&T_n^*=\left(K+\frac{c_{2,K}}{c_{1,K}}N\right)\frac{\sqrt{a_{n}}}{\sum_{n\in \mathcal N}\sqrt{a_{n}}}-\frac{c_{2,K}}{c_{1,K}}>0, \ n\in \mathcal N,\label{eqn:sum-pi-n}\\
&q^*_{K,\infty}=\frac{1}{c_{1,K}}\left(1-\frac{(\sum_{n\in \mathcal N}\sqrt{a_{n}})^2}{N+K\frac{c_{1,K}}{c_{2,K}}}\right).
\label{eqn:opt-q-K}
\end{align}
Here, $c_{1,K}$ and $c_{2,K}$ are given by \eqref{eqn:c1-def} and \eqref{eqn:c2-def}, respectively.
\label{Cor:solu-opt-K-infty}
\end{Cor}

Note that Corollary~\ref{Cor:solu-opt-K-infty} can be interpreted in a similar way to   Corollary~\ref{Cor:solu-opt-1-infty}.
In addition, by Corollary~\ref{Cor:solu-opt-1-infty} and Corollary~\ref{Cor:solu-opt-K-infty}, we can see that under certain conditions, $q^*_{K,\infty}$ increases  with $K$. This indicates the tradeoff between the network performance and network  storage resources.

Based on Lemma \ref{Lem:equivalence}, we know that  any $\mathbf p^*$ in the convex polyhedron $\{\mathbf p^*: \eqref{eqn:def-T-n*}, \eqref{eqn:cache-constr-indiv}, \eqref{eqn:cache-constr-sum}\}$  is an optimal solution to Problem~\ref{prob:opt-K-infty-p}, i.e., an asymptotically optimal solution to Problem~\ref{prob:opt} for $K>1$, where
\begin{align}
\sum_{i\in \mathcal I_n}p_i^*=T_n^*, \ n\in \mathcal N.\label{eqn:def-T-n*}
\end{align}
and $\mathbf T^*$ is given by Theorem~\ref{Thm:solu-opt-K-infty-t} or Corollary~\ref{Cor:solu-opt-K-infty}. In other words, we have a set of asymptotically optimal solutions in the high SNR and user density region.

Next, we obtain the best asymptotically optimal solution which achieves the optimal successful transmission probability  in the general region among the set of asymptotically optimal solutions, i.e., $\{\mathbf p^*: \eqref{eqn:def-T-n*}, \eqref{eqn:cache-constr-indiv}, \eqref{eqn:cache-constr-sum}\}$. Specifically, substituting $\mathbf p^*$ satisfying  \eqref{eqn:def-T-n*} into $q_K\left(\mathbf{p}\right)$ in \eqref{eq:CPrate_multifile_noise}, we have
 \begin{align}
 q_K(\mathbf p^*)=&\sum\limits_{n\in \mathcal N}\frac{a_{n}}{T_n^*} \sum_{i\in \mathcal I_n}p_i^*\left(\sum_{k=1}^K\sum_{\mathcal N_i^1 \in \mathcal {SN}_i^1(k-1)}G_{n,i}(\mathcal N_i^1,\mathbf T^*_{i,-n}) f_k(T_n^*)\right)\nonumber\\
 =&\sum_{i\in \mathcal I}\left(\sum_{n\in \mathcal N_i}\frac{a_{n}}{T_n^*} \sum_{k=1}^K\sum_{\mathcal N_i^1 \in \mathcal {SN}_i^1(k-1)}G_{n,i}(\mathcal N_i^1,\mathbf T^*_{i,-n}) f_k(T_n^*)\right)p_i^*\triangleq q_K(\mathbf p^*,\mathbf T^*).\nonumber
 \end{align}
For given $\mathbf T^*$, we would like to obtain the best asymptotically optimal solution which maximizes the successful transmission probability $q_K(\mathbf p^*,\mathbf T^*)$  in the general region among all the asymptotically optimal solutions in  $\{\mathbf p^*: \eqref{eqn:def-T-n*}, \eqref{eqn:cache-constr-indiv}, \eqref{eqn:cache-constr-sum}\}$.
\begin{Prob} [Optimization of $\mathbf p^*$ under given $\mathbf T^*$ for $K>1$]
\begin{align}\overline{q^*_{K,\infty}}\triangleq\max_{\mathbf p^*} &\quad q_K(\mathbf p^*,\mathbf T^*)\nonumber\\
s.t. & \quad \eqref{eqn:def-T-n*}, \eqref{eqn:cache-constr-indiv}, \eqref{eqn:cache-constr-sum}.\nonumber
\end{align}\label{Prob:asm-improvement}
\end{Prob}

Problem~\ref{Prob:asm-improvement} is a linear programming problem. The optimal solution to Problem~\ref{Prob:asm-improvement}  can be obtained using the simplex method.
To further reduce complexity, before applying the simplex method, we can first derive some caching probabilities which are zero based on the relationship between $\mathbf p^*$ and $\mathbf T^*$ in \eqref{eqn:def-T-n*}.
In particular, for all $i\in\mathcal I_n$ and $n\in\{n\in\mathcal N: T_n^*=0\}$, we have $p_i^*=0$; for all $i\not\in\mathcal I_n$ and $n\in\{n\in\mathcal N: T_n^*=1\}$, we have $p_i^*=0$. Thus, we have
\begin{align}
p_i^*=0, \ i\in\mathcal I', \label{eqn:p_i-0}
\end{align}
where  $\mathcal I'\triangleq \cup_{n\in\{n\in\mathcal N: T_n^*=0\}}\mathcal I_n\cup\left(\mathcal I\setminus \cup_{n\in\{n\in\mathcal N: T_n^*=1\}}\mathcal I_n\right)$.
Then, we can compute  the remaining caching probabilities for the communications in $\mathcal I\setminus \mathcal I'$ using the simplex method.

Therefore, by solving Problem~\ref{prob:opt-K-infty} and Problem~\ref{Prob:asm-improvement},  we can obtain the best asymptotically optimal solution in the general region, using the method  summarized in Algorithm~\ref{alg:sympK}.\footnote{Note that  Algorithm~\ref{alg:sympK} is different from   the method demonstrated in Fig. 1 of \cite{ICC15Giovanidis}.} Note that the computational complexity of Algorithm~\ref{alg:sympK} is much smaller than that of Algorithm~\ref{alg:localK} due to the following reasons: (i) since $\mathbf T^*$ from Theorem~\ref{Thm:solu-opt-K-infty-t} (Corollary~\ref{Cor:solu-opt-K-infty}) is in closed-form, the complexity of Step~1 of Algorithm~\ref{alg:sympK} is low; (ii) for most popularity distributions of interest, the cardinality of set $\{n\in\mathcal N: T_n^*>0\}$ is much smaller than $N$, and hence, the number of combinations of possibly positive caching probability (i.e., $I-|\mathcal I'|$) is actually much smaller than the total number of combinations (i.e., $I$); and (iii) Step~3 of Algorithm~\ref{alg:sympK} is of manageable complexity.

\begin{algorithm} [Local Optimal Solution for $K>1$]
\caption{Asymptotically Optimal Solution for $K>1$}
\begin{algorithmic}[1]
\STATE Obtain the optimal solution  $\mathbf T^*$  to Problem~\ref{prob:opt-K-infty} based on Theorem~\ref{Thm:solu-opt-K-infty-t} or Corollary~\ref{Cor:solu-opt-K-infty}.
\STATE Determine $\mathcal I'$ and choose $p_i^*=0$ for all $i\in \mathcal I'$ according to \eqref{eqn:p_i-0}.
 \STATE  Obtain $\{p_i^*: i\in \mathcal I\setminus \mathcal I'\}$ by solving Problem~\ref{Prob:asm-improvement} under the constraint in  \eqref{eqn:p_i-0} using  the simplex method.
\end{algorithmic}\label{alg:sympK}
\end{algorithm}

Now, we compare the proposed local optimal design $\mathbf p^{\dagger}$ (Local Opt. obtained using  Algorithm~\ref{alg:localK}) and the best asymptotically optimal  design $\mathbf p^*$ (Asymp. Opt. obtained using  Algorithm~\ref{alg:sympK}) using a numerical example.
From Fig.~\ref{fig:comparison_local_linear}, we can see that the performance of Asymp. Opt. is very close to that of Local Opt., even at moderate transmit SNR  and user density.
On the other hand,  the average matlab computation time for $\mathbf p^{\dagger}$ is 112 times of that for $\mathbf p^*$, indicating that
the complexity of Local Opt. is much higher than that of  Asymp. Opt. These demonstrate the  applicability and effectiveness  of the proposed   Asymp. Opt. in the general transmit SNR and user density region.

\begin{figure}[h]
\begin{center}
 \includegraphics[width=7cm]{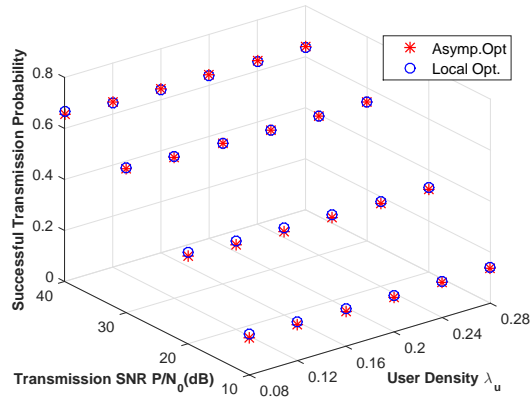}
  \end{center}
    \caption{\small{Comparison between Local Opt. and Asymp. Opt. $N=8$, $K=4$, $\alpha = 4$, $\lambda_b = 0.01$,  $W = 10\times 10^6$, $\tau = 5\times10^5$, and $a_n=\frac{n^{-\gamma}}{\sum_{n\in \mathcal N}n^{-\gamma}}$  with $\gamma =0.8$.}}
\label{fig:comparison_local_linear}
\end{figure}

%

\section{Numerical Results}\label{Sec:simu}

\begin{figure}[t]
\begin{center}
  \subfigure[\small{
Cache size  at $\gamma=0.6$, $\lambda_b=0.02$ and $\lambda_u=0.1$.}]
  {\resizebox{7cm}{!}{\includegraphics{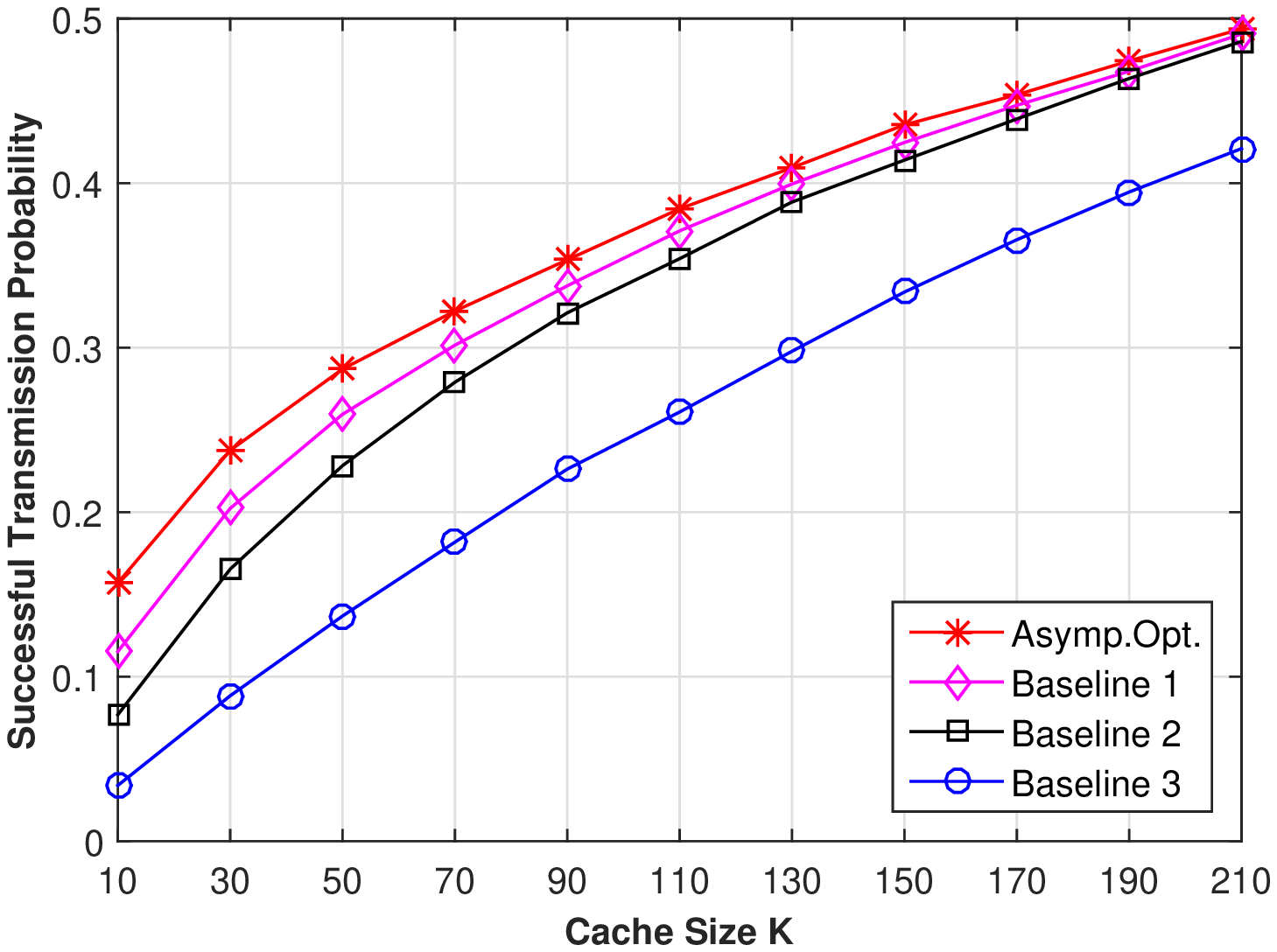}}}\quad
  \subfigure[\small{
 Zipf exponent at $K=30$, $\lambda_b=0.02$ and $\lambda_u=0.1$.}]
  {\resizebox{7cm}{!}{\includegraphics{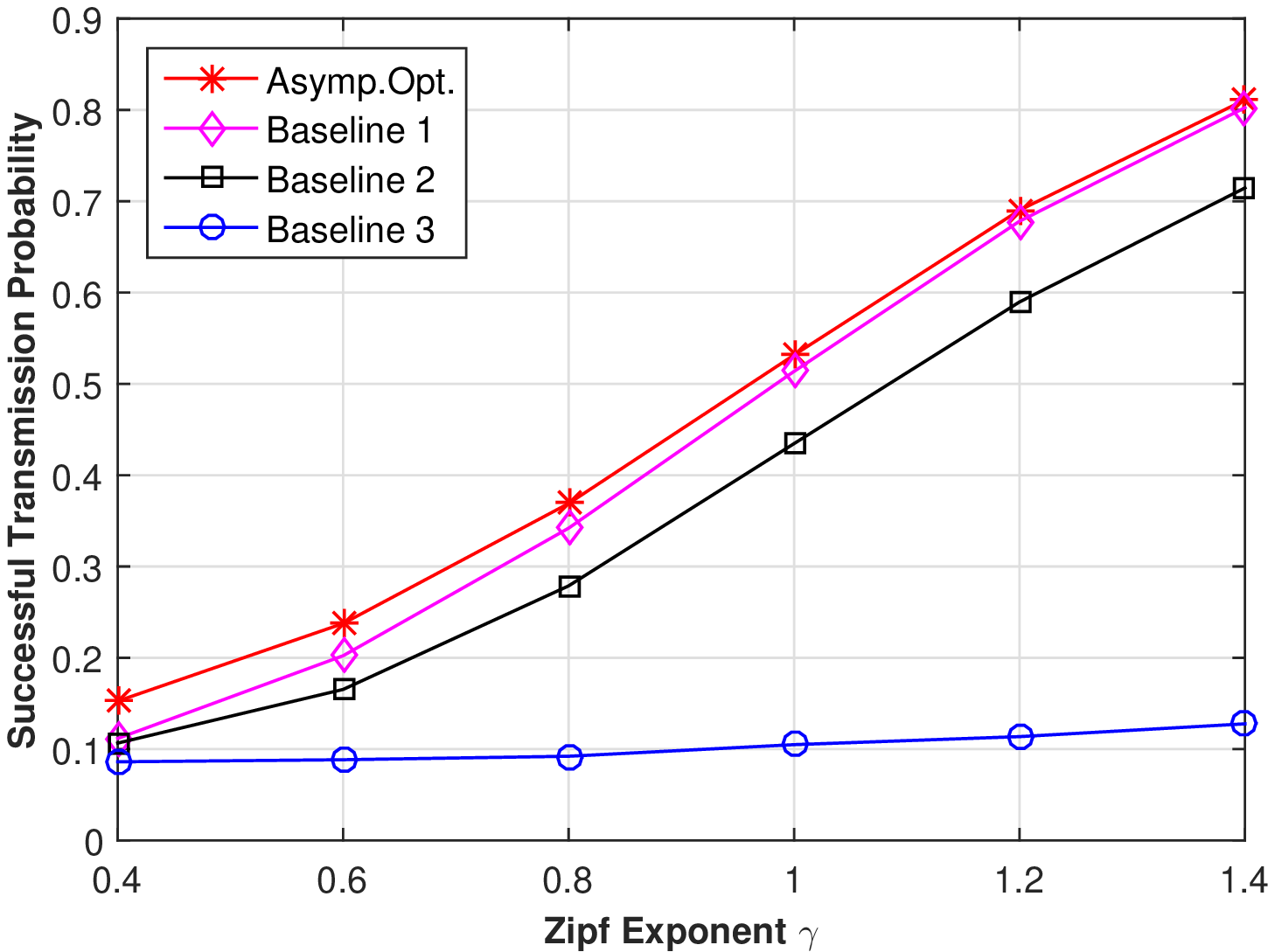}}}\quad
  \subfigure[\small{
BS density  at $K=30$, $\gamma=0.6$ and $\lambda_u=0.1$.}]
  {\resizebox{7cm}{!}{\includegraphics{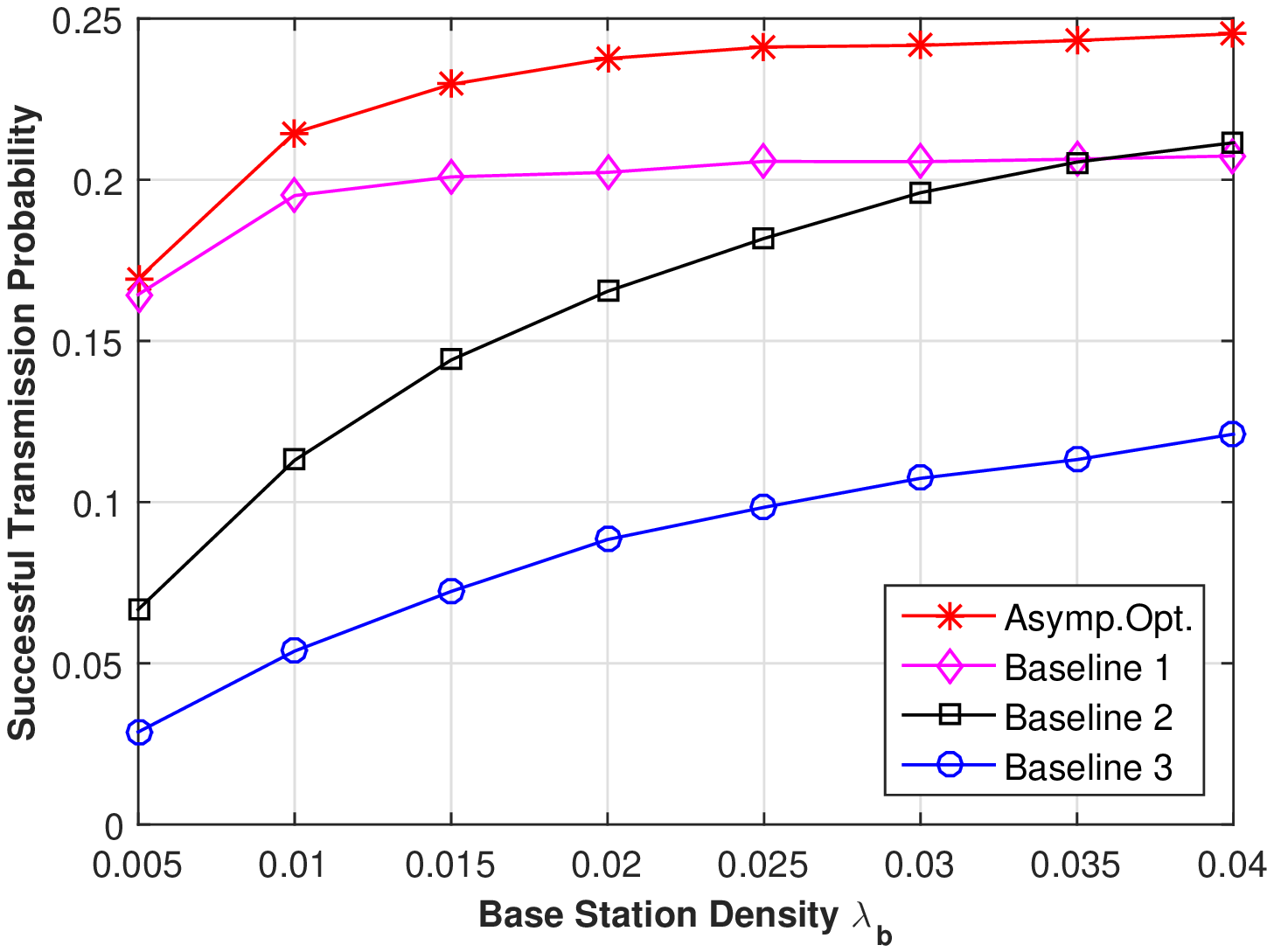}}}
  \subfigure[\small{
User density  at $K=30$, $\gamma=0.6$ and $\lambda_b =0.02$.}]
  {\resizebox{7cm}{!}{\includegraphics{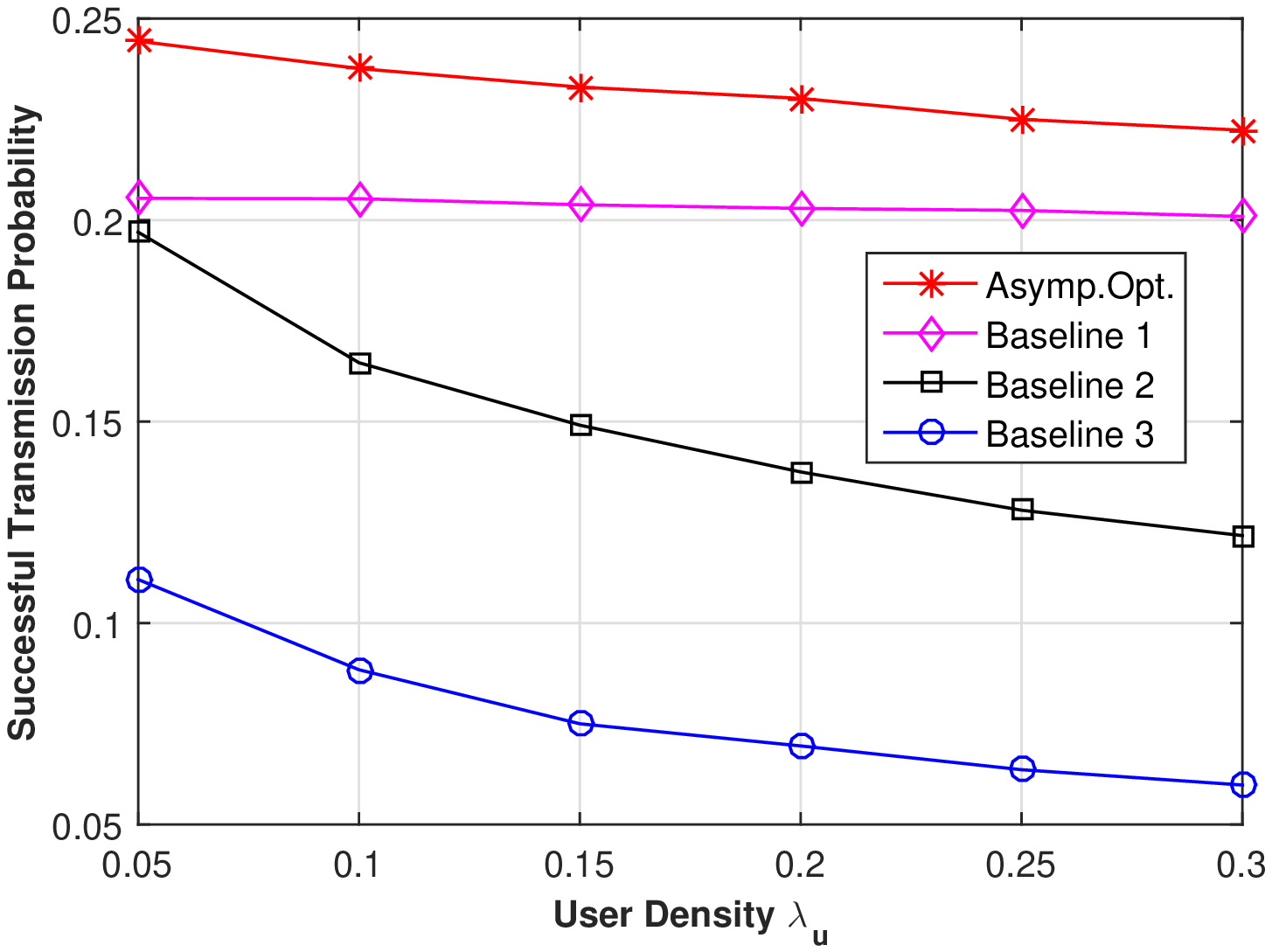}}}
  \end{center}
    \caption{\small{Successful transmission probability versus cache size $K$, Zipf exponent $\gamma$, base station density $\lambda_b$ and user density $\lambda_u$ for $\frac{P}{N_0}=30dB$,  $\alpha=4$, $W = 10\times 10^6$, $\tau =  10^5$ and $N = 1000$.}}
\label{fig:simulation-large}
\end{figure}

In this section, we compare the proposed asymptotically optimal  design $\mathbf p^*$ with three  schemes.\footnote{When $N=1000$, the complexity of Local Opt. is not acceptable. Thus, in Fig.~\ref{fig:simulation-large}, we only consider the proposed Asymp. Opt. and the three baseline schemes.}
Baseline 1 refers to the design in which  the most $K$ popular  files are stored at each BS, i.e., $T_n=1$ for all $n=1,\cdots, K$ and $T_n=0$ for all $n=K+1,\cdots, N$ \cite{ICC15Giovanidis,EURASIP15Debbah}. Baseline 2 refers to the design in which each BS  selects $K$ files in an i.i.d. manner with  file $n$ being selected with probability $a_n$ \cite{DBLP:journals/corr/BharathN15}. Note that in Baseline 2, each BS may cache multiple copies of one file, leading to storage  waste.
Baseline 3 refers to the design in which each BS randomly selects one combination to cache according to the uniform distribution \cite{cachingmimoLiu15}, i.e., $p_i=\frac{1}{I}$ for all $i\in \mathcal I$.
Note that the three baseline schemes also adopt the multicasting scheme as in our design.
In the simulation, we assume the popularity follows Zipf distribution, i.e., $a_n=\frac{n^{-\gamma}}{\sum_{n\in \mathcal N}n^{-\gamma}}$, where $\gamma$ is the Zipf exponent. Note that a small $\gamma$ means a heavy-tail popularity distribution.

Fig.~\ref{fig:simulation-large} illustrates   the successful transmission probability versus different parameters. From Fig.~\ref{fig:simulation-large}, we can observe that the proposed design outperforms all the three baseline schemes. In addition, the proposed design, Baseline 1 and Baseline 2 have much better performance than Baseline 3, as they exploit file popularity to improve the performance. The performance gap between the proposed design and Baseline 1 is relatively large at small cache size $K$, small Zipf exponent $\gamma$, large BS density $\lambda_b$ and small user density $\lambda_u$. This demonstrates the benefit of file diversity in these regions.
Please note that the successful transmission probability in Fig.~\ref{fig:simulation-large} is small due to the low file availability when $K/N$ is small. This performance is achieved without backhaul cost or extra delay cost. Our goal here is to study how cache itself can affect the system performance.

Specifically,
Fig.~\ref{fig:simulation-large} (a)  illustrates   the successful transmission probability versus  the cache size $K$.
We can see that the performance of all the schemes increases with $K$. This is because as $K$ increases, each BS can cache more files, and the probability that a randomly requested file is cached at a nearby BS increases. Fig.~\ref{fig:simulation-large} (b)  illustrates   the  successful transmission probability versus  the Zipf exponent $\gamma$. We can observe that the performance of the proposed design, Baseline 1 and Baseline 2  increases with the Zipf exponent $\gamma$ much faster than Baseline 3. This is because when $\gamma$ increases, the tail of popularity distribution becomes small, and hence, the average network file load decreases. The performance increase of Baseline 3 with $\gamma$ only comes from the decrease of the average network file load. While, under  the proposed design, Baseline 1 and Baseline 2, the probability that a randomly requested file is cached at a nearby BS increases with $\gamma$. Thus,
the performance increases of the proposed design, Baseline 1 and Baseline 2 with $\gamma$ are due to  the decrease of the  average network file load and the increase of the chance of a requested file being cached at a nearby BS.  Fig.~\ref{fig:simulation-large} (c)  illustrates   the successful transmission probability versus  the BS density $\lambda_b$. We can see that the performance of all the schemes increases with $\lambda_b$. This is because the distance between a user whose requested file is cached in the network and its serving BS decreases, as $\lambda_b$ increases. Fig.~\ref{fig:simulation-large} (d)  illustrates   the successful transmission probability versus  the user density $\lambda_u$. We can see that the performance of all the schemes decreases with $\lambda_u$. This is because the probability of a cached file being requested by at least one user  increases, as $\lambda_u$ increases. In addition,
 in Fig.~\ref{fig:simulation-large} (c) The performance of the proposed design and Baseline 1 increases slower than that of Baseline 2 and Baseline 3; in Fig.~\ref{fig:simulation-large} (d), the performance of the proposed design and Baseline 1 decreases slower than that of Baseline 2 and Baseline 3. The reason is that more different files are cached in the network under Baseline 2 and Baseline 3 than under the proposed design and Baseline 1. Note that  under the proposed design, usually more than $K$ but less than $N$ different files are cached in the network; under Baseline 1, only $K$ most popular files are cached in the network; under Baseline 2 and Baseline 3, all $N$ files are cached in the network.

\section{Conclusion}
 In this paper, we consider  the analysis and optimization of caching and multicasting  in a large-scale  cache-enabled wireless  network. We propose a random caching and multicasting scheme with a design parameter. Utilizing tools from stochastic geometry, we first derive  the successful transmission probability. Then, using optimization techniques, we develop an iterative  numerical algorithm to obtain a local optimal design in the general   region, and also derive a simple asymptotically optimal design  in the high SNR  and user density region.
 Finally, we show that the asymptotically  optimal design achieves promising performance in the general  region with much lower complexity than the local optimal design.


\section*{Appendix A: Proof of Theorem~\ref{Thm:generalK1}}

First, as illustrated in Section~\ref{Subsec:K1-analysis}, when $K=1$, there are two types of interferers, i.e., the interfering BSs storing the file requested  by $u_{0}$   and the interfering BSs storing the other files.   Thus, we rewrite the SINR expression ${\rm SINR}_{n,0}$ in \eqref{eqn:SINR} as
\begin{align}\label{eq:SINR_K1_v2}
{\rm SINR}_{n,0} 
&=\frac{{D_{0,0}^{-\alpha}}\left|h_{0,0}\right|^{2}}{I_{n}+\sum_{m\in\mathcal{N},m\neq n}I_{m}+\frac{N_{0}}{P}},
\end{align}
\normalsize{
where $\Phi_{b,n}$ ($n\in \mathcal N$) denotes the point process generated by BSs storing file $n$, $I_{n}\triangleq\sum_{\ell\in\Phi_{b,n}\backslash B_{n,0}}D_{\ell,0}^{-\alpha}\left|h_{\ell,0}\right|^{2}$, and $I_{m}\triangleq\sum_{\ell\in\Phi_{b,m}}D_{\ell,0}^{-\alpha}\left|h_{\ell,0}\right|^{2}$ ($m\in\mathcal{N}, m\neq n$). Due to the random caching policy and independent thinning \cite[Page 230]{FTNhaenggi09}, we know that $\Phi_{b,m}$ is a homogeneous PPP with density $p_{m}\lambda_{b}$.}

Next, we calculate the conditional successful transmission probability of file $n$ requested by $u_{0}$ conditioned on $D_{0,0}=d$, denoted as   $q_{1,n,D_{0,0}}\left({\bf p},d\right)\triangleq {\rm Pr}\left[W\log_{2}\left(1+{\rm SINR}_{n,0}\right)\ge \tau\big|D_{0,0}=d\right]$. Based on \eqref{eq:SINR_K1_v2}, we have
\begin{align}
&q_{1,n,D_{0,0}}\left({\bf p},d\right)\notag\\
\eqla&{\rm E}_{I_{1},\ldots,I_{N}}\left[{\rm Pr}\left[\left|h_{0,0}\right|^{2}\ge \left(2^{\frac{\tau}{W}}-1\right)D_{0,0}^{\alpha}\left(I_{n}+\sum_{m\in\mathcal N,m\neq n}^{N}I_{m}+\frac{N_{0}}{P}\right)\Big|D_{0,0}=d\right]\right]\notag\\
\eqlb&{\rm E}_{I_{1},\ldots,I_{N}}\left[\exp\left(-\left(2^{\frac{\tau}{W}}-1\right)d^{\alpha}\left(I_{n}+\sum_{m\in\mathcal{N},m\neq n}I_{m}+\frac{N_{0}}{P}\right)\right)\right]\notag\\
\eqlc&\underbrace{{\rm E}_{I_{n}}\left[\exp\left(-\left(2^{\frac{\tau}{W}}-1\right)d^{\alpha}I_{n}\right)\right]}_{\triangleq\mathcal{L}_{I_{n}}(s,d)|_{s=\left(2^{\frac{\tau}{W}}-1\right)d^{\alpha}} }\prod_{m\in\mathcal{N},m\neq n}\underbrace{{\rm E}_{I_{m}}\left[\exp\left(-\left(2^{\frac{\tau}{W}}-1\right)d^{\alpha}I_{m}\right)\right]}_{\triangleq\mathcal{L}_{I_{m}}(s,d)|_{s=\left(2^{\frac{\tau}{W}}-1\right)d^{\alpha}} } \nonumber\\
&\times\exp\left(-\left(2^{\frac{\tau}{W}}-1\right)d^{\alpha}\frac{N_{0}}{P}\right),\label{eq:condi_CP_1}
\end{align}
\normalsize{where $(a)$ is obtained based on (\ref{eq:SINR_K1_v2}), (b) is obtained by noting that $|h_{0,0}|^{2}\dis \exp(1)$, and (c) is due to the independence of the Rayleigh fading channels and the independence of the PPPs $\Phi_{m}$ ($m\in\mathcal{N}$). To calculate $q_{1,n,D_{0,0}}\left({\bf p},d\right)$ according to \eqref{eq:condi_CP_1}, we first calculate $\mathcal{L}_{I_{n}}(s,d)$ and $\mathcal{L}_{I_{m}}(s,d)$ ($m\in \mathcal N, m\neq n$), respectively. The expression of $\mathcal{L}_{I_{n}}(s,d)$ is calculated below.}
\begin{align}\label{eq:LT_K1_n}
\mathcal{L}_{I_{n}}(s,d)=&{\rm E}\left[\exp\left(-s\sum_{\ell\in\Phi_{b,n}\backslash B_{n,0}}D_{\ell,0}^{-\alpha}\left|h_{\ell,0}\right|^{2} \right)\right]
={\rm E}\left[\prod_{\ell\in\Phi_{b,n}\backslash B_{n,0}}\exp\left(-s D_{\ell,0}^{-\alpha}\left|h_{\ell,0}\right|^{2} \right)\right]\notag\\
\eqld&\exp\left(-2\pi p_{n}\lambda_{b}\int_{d}^{\infty}\left(1-\frac{1}{1+sr^{-\alpha}}\right)r{\rm d}r\right)\nonumber\\
\eqle&\exp\left(-\frac{2\pi}{\alpha}p_{n}\lambda_{b}s^{\frac{2}{\alpha}}B^{'}\left(\frac{2}{\alpha},1-\frac{2}{\alpha},\frac{1}{1+sd^{-\alpha}}\right)\right),
\end{align}
\normalsize{where $(d)$ is obtained by utilizing the probability generating functional of PPP \cite[Page 235]{FTNhaenggi09}, and $(e)$ is obtained by first replacing $s^{-\frac{1}{\alpha}}r$ with $t$, and then replacing $\frac{1}{1+t^{-\alpha}}$ with $w$. Similarly, the expression of $\mathcal{L}_{I_{m}}(s,d)$ ($m\in\mathcal{N}, m\neq n$) is calculated as follows:}
\begin{align}\label{eq:LT_K1_k}
\mathcal{L}_{I_{m}}(s,d)=&{\rm E}\left[\exp\left(-s\sum_{\ell\in\Phi_{b,m}}D_{\ell,0}^{-\alpha}\left|h_{\ell,0}\right|^{2}\right)\right]={\rm E}\left[\prod_{\ell\in\Phi_{b,m}}\exp\left(-s D_{\ell,0}^{-\alpha}\left|h_{\ell,0}\right|^{2}\right)\right]\notag\\
=&\exp\left(-2\pi p_{m}\lambda_{b}\int_{0}^{\infty}\left(1-\frac{1}{1+sr^{-\alpha}}\right)r{\rm d}r\right)\nonumber\\
=&\exp\left(-\frac{2\pi}{\alpha}p_{m}\lambda_{b}s^{\frac{2}{\alpha}}B\left(\frac{2}{\alpha},1-\frac{2}{\alpha}\right)\right)\;.
\end{align}
\normalsize{Substituting (\ref{eq:LT_K1_n}) and (\ref{eq:LT_K1_k}) into (\ref{eq:condi_CP_1}), we obtain   $q_{1,n,D_{0,0}}\left({\bf p},d\right)$ as follows:}
\begin{align}\label{eq:condi_CP_1_v2}
&q_{1,n,D_{0,0}}\left({\bf p},d\right)\nonumber\\
=&\exp\left(-\frac{2\pi}{\alpha}p_{n}\lambda_{b}d^{2}\left(2^{\frac{\tau}{W}}-1\right)^{\frac{2}{\alpha}}B^{'}\left(\frac{2}{\alpha},1-\frac{2}{\alpha},2^{-\frac{\tau}{W}}\right)\right)\exp\left(-\left(2^{\frac{\tau}{W}}-1\right)d^{\alpha}\frac{N_{0}}{P}\right)\notag\\
&\times  \exp\left(-\frac{2\pi}{\alpha}(1-p_n)\lambda_{b}d^{2}\left(2^{\frac{\tau}{W}}-1\right)^{\frac{2}{\alpha}}B\left(\frac{2}{\alpha},1-\frac{2}{\alpha}\right)\right).
\end{align}

\normalsize{Now, we calculate $q_{1,n}\left({\bf p}\right)$ by removing the condition of $q_{1,n,D_{0,0}}\left({\bf p},d\right)$ on $D_{0,0}=d$.  Note that we have the p.d.f. of $D_{0,0}$ as $f_{D_{0,0}}(d)=2\pi p_{n}\lambda_{b}d\exp\left(-\pi p_{n}\lambda_{b}d^{2}\right)$, as the BSs storing file $n$ form a homogeneous PPP with density $p_{n}\lambda_{b}$. Thus, we have:}
\begin{align}
&q_{1,n}\left({\bf p}\right)=\int_{0}^{\infty}q_{1,n,D_{0,0}}\left({\bf p},d\right)f_{D_{0,0}}(d){\rm d}d=f_1(p_n),\notag
\end{align}
\normalsize{where $f_1(p_n)$ is given by \eqref{eqn:def-f}.
Finally, by $q_1(\mathbf p)=\sum_{n\in \mathcal N}a_{n}q_{1,n}(\mathbf p)$, we can prove Theorem~\ref{Thm:generalK1}.}

\section*{Appendix B: Proof of Corollary~\ref{Cor:generalK1}}

When $\frac{P}{N_{0}}\to\infty$, $\exp\left(-\left(2^{\frac{\tau}{W}}-1\right)d^{\alpha}\frac{N_{0}}{P}\right)\to1$. Thus, by \eqref{eq:CPrate_1file_noise}, we have:
\begin{align}
q_{1,\infty}\left(\mathbf{p}\right)=&2\pi\lambda_{b}\sum_{n\in \mathcal N}a_{n}p_{n}\int_{0}^{\infty}d\exp\left(-\frac{2\pi}{\alpha}p_{n}\lambda_{b}\left(2^{\frac{\tau}{W}}-1\right)^{\frac{2}{\alpha}}d^{2}B^{'}\left(\frac{2}{\alpha},1-\frac{2}{\alpha},2^{-\frac{\tau}{W}}\right)\right)\notag\\
&\times\exp\left(-\pi\lambda_{b}p_{n}d^{2}\right)\exp\left(-\frac{2\pi}{\alpha}\left(1-p_{n}\right)\lambda_{b}\left(2^{\frac{\tau}{W}}-1\right)^{\frac{2}{\alpha}}d^{2}B\left(\frac{2}{\alpha},1-\frac{2}{\alpha}\right)\right){\rm d}d\notag\\
=&2\pi\lambda_{b}\sum_{n\in \mathcal N}a_{n}p_{n}\int_{0}^{\infty}d\exp\left(-\pi\lambda_{b}p_{n}\left(c_{2,1}+c_{2,1}\frac{1}{p_{n}}\right)d^{2}\right){\rm d}d\nonumber
\end{align}
\normalsize{Noting that $\int_{0}^{\infty}d\exp\left(-cd^{2}\right){\rm d}d=\frac{1}{2c}$ ($c$ is a constant), we can solve the integral and prove Corollary~\ref{Cor:generalK1}.}

\section*{Appendix C: Proof of Theorem~\ref{Thm:solu-opt-1-infty}}

The Lagrangian of Problem \ref{prob:opt-1-infty} is given by
$$L(\mathbf p, \boldsymbol \eta, \nu)=\sum_{n\in \mathcal N}\frac{a_np_n}{c_{2,1}+c_{1,1}p_n}+\sum_{n\in \mathcal N}\eta_np_n+\nu\left(1-\sum_{n\in \mathcal N} p_n\right),$$
\normalsize{where $\eta_n\geq 0$ is the Lagrange multiplier associated with \eqref{eqn:cache-constr-indiv},  $\nu$ is the Lagrange multiplier associated with \eqref{eqn:cache-constr-sum}, and $\boldsymbol\eta\triangleq (\eta_n)_{n\in \mathcal N}$. Thus, we have
$\frac{\partial L}{\partial  p_n}(\mathbf p, \boldsymbol \eta, \nu)=\frac{a_nc_{2,1}}{(c_{2,1}+c_{1,1}p_n)^2}+\eta_n-\nu.$  Since strong duality holds, primal optimal $\mathbf p^*$ and dual optimal  $\boldsymbol\eta^*$, $\nu^*$ satisfy KKT conditions, i.e.,  (i) primal constraints: \eqref{eqn:cache-constr-indiv}, \eqref{eqn:cache-constr-sum},  (ii) dual constraints $\eta_n\geq 0$ for all $n\in \mathcal N$, (iii) complementary slackness $\eta_np_n=0$ for all $n\in \mathcal N$, and (iv) $\frac{a_nc_{2,1}}{(c_{2,1}+c_{1,1}p_n)^2}+\eta_n-\nu=0$ for all $n\in \mathcal N$. By \eqref{eqn:cache-constr-indiv}, (ii), (iii) and (iv), we have:  if $\nu<\frac{a_n}{c_{2,1}}$, then $\eta_n=0$ and $p_n=\frac{1}{c_{1,1}}\sqrt{\frac{a_nc_{2,1}}{\nu}}-\frac{c_{2,1}}{c_{1,1}}$; if  $\nu\geq\frac{a_n}{c_{2,1}}$, then $\eta_n=\nu-\frac{a_nc_{2,1}}{(c_{2,1}+c_{1,1}p_n)^2}$ and $p_n=0$. Thus, we have $p_n^*=\left[\frac{1}{c_{1,1}}\sqrt{\frac{a_nc_{2,1}}{\nu^*}}-\frac{c_{2,1}}{c_{1,1}}\right]^+$. Combining   \eqref{eqn:cache-constr-sum}, we can prove Theorem~\ref{Thm:solu-opt-1-infty}.}

\section*{Appendix D: Proof of Lemma~\ref{Lem:pmf-K}}

Let the random variable $Y_{m,n,i}\in\{0,1\}$ denote whether file $m\in \mathcal N_{i,-n}$ is requested by the users associated with $B_{n,0}$ when  $B_{n,0}$ contains combination $i\in \mathcal I_n$. When $B_{n,0}$ contains combination $i\in \mathcal I_n$, we have $K_{n,0}=1+\sum_{m\in \mathcal N_{i,-n}} Y_{m,n,i}$. Thus, we have
\begin{align}
&\Pr \left[K_{n,0}=k|\text{$B_{n,0}$ contains combination $i\in \mathcal I_n$}\right]\nonumber\\
=&\sum_{\mathcal N_i^1\in \mathcal{SN}_i^1(k-1) }\prod_{m\in \mathcal N_i^1}(1-\Pr[Y_{m,n,i}=0])\prod_{m\in \mathcal N_{i,-n}\setminus\mathcal N_i^1}\Pr[Y_{m,n,i}=0],\label{eqn:K-pmf-i}
\end{align}
\normalsize{where $k=1,\cdots, K$. The probability that $B_{n,0}$ contains combination $i\in \mathcal I_n$ is $\frac{p_i}{T_n}$. Thus, by the law of total probability, we have}
\begin{align}
&\Pr \left[K_{n,0}=k\right]
=\sum_{i\in \mathcal I_n}\frac{p_i}{T_n}\Pr \left[K_{n,0}=k|\text{$B_{n,0}$ contains combination $i\in \mathcal I_n$}\right],\nonumber
\end{align}
\normalsize{where $k=1,\cdots, K$. Thus, to prove \eqref{eqn:K-pmf}, it remains to  calculate $\Pr[Y_{m,n,i}=0]$.   The p.m.f.  of  $Y_{m,n,i}$ depends on the  p.d.f. of the size of the Voronoi cell of $B_{n,0}$ w.r.t. file $m\in \mathcal N_{i,-n}$  when  $B_{n,0}$ contains combination $i\in \mathcal I_n$, which is unknown. We approximate this p.d.f. based on  the  p.d.f. of the size of the Voronoi cell to which a randomly chosen user belongs \cite{SGcellsize13}.
Under this approximation, we can calculate the p.m.f. of  $Y_{m,n,i}$ using Lemma 3 of \cite{SGcellsize13}: $Pr[Y_{m,n,i}=0]
\approx\left(1+3.5^{-1}\frac{a_m\lambda_u}{T_m\lambda_b}\right)^{-4.5}$.
Therefore, we complete the proof.}

\section*{Appendix E: Proof of Theorem~\ref{Thm:generalKmulti}}


First, as illustrated in Section~\ref{Subsec:K-analysis}, when $K>1$, there are two types of interferers, i.e., the interfering BSs storing the combinations containing the file requested by $u_{0}$  and the interfering BSs without  the desired file of $u_0$
Thus, we rewrite the SINR expression ${\rm SINR}_{n,0}$ in \eqref{eqn:SINR} as:
\begin{align}\label{eq:SINR_K>1_v2}
{\rm SINR}_{n,0} 
=\frac{{D_{0,0}^{-\alpha}}\left|h_{0,0}\right|^{2}}{I_{n}+I_{-n}+\frac{N_{0}}{P}},
\end{align}
\normalsize{where $\Phi_{b,n}$ is the point process generated by BSs containing file combination $i\in\mathcal{I}_{n}$ and $\Phi_{b,-n}$ is the point process generated by BSs containing file combination $i\in \mathcal I, i\not\in\mathcal{I}_{n}$, $I_n\triangleq\sum_{\ell\in\Phi_{b,n}\backslash B_{n,0}}D_{\ell,0}^{-\alpha}\left|h_{\ell,0}\right|^{2}$, and $I_{-n}\triangleq \sum_{\ell\in\Phi_{b,-n}}D_{\ell,0}^{-\alpha}\left|h_{\ell,0}\right|^{2}$. Due to the random caching policy and independent thinning \cite[Page 230]{FTNhaenggi09}, we obtain that $\Phi_{b,n}$ is an homogeneous PPP with density $\lambda_{b}T_n$ and $\Phi_{b,-n}$ is an homogeneous PPP with density $\lambda_{b}\left(1-T_n\right)$.}

Next, we calculate the conditional successful transmission probability of file $n$ requested by $u_{0}$ conditioned on $D_{0,0}=d$ when the file load is $k$, denoted as   $$q_{k,n,D_{0,0}}\left({\bf p},d\right)\triangleq {\rm Pr}\left[\frac{W}{k}\log_{2}\left(1+{\rm SINR}_{n,0}\right)\ge \tau\big|D_{0,0}=d\right].$$ \normalsize{Similar to (\ref{eq:condi_CP_1}) and based on (\ref{eq:SINR_K>1_v2}), we have:}
\begin{align}\label{eq:condi_CP_K}
&q_{k,n,D_{0,0}}\left({\bf p},d\right)\nonumber\\
=&{\rm E}_{I_n,I_{-n}}\left[{\rm Pr}\left[\left|h_{0,0}\right|^{2}\ge \left(2^{\frac{k\tau}{W}}-1\right)D_{0,0}^{\alpha}\left(I_n+I_{-n}+\frac{N_{0}}{P}\right)\Big|D_{0,0}=d\right]\right]\notag\\
=&\mathcal{L}_{I_{n}}(s,d)|_{s=\left(2^{\frac{k\tau}{W}}-1\right)d^{\alpha}} \mathcal{L}_{I_{-n}}(s,d)|_{s=\left(2^{\frac{k\tau}{W}}-1\right)d^{\alpha}}\exp\left(-\left(2^{\frac{k\tau}{W}}-1\right)d^{\alpha}\frac{N_{0}}{P}\right)\;.
\end{align}
\normalsize{To calculate $q_{k,n,D_{0,0}}\left({\bf p},d\right)$ according to \eqref{eq:condi_CP_K}, we first calculate $\mathcal{L}_{I_{n}}(s,d)$ and $\mathcal{L}_{I_{-n}}(s,d)$, respectively. Similar to (\ref{eq:LT_K1_n}) and (\ref{eq:LT_K1_k}), we have}
\begin{align}
\mathcal{L}_{I_n}(s,d)=&\exp\left(-\frac{2\pi}{\alpha}T_n\lambda_{b}s^{\frac{2}{\alpha}}B^{'}\left(\frac{2}{\alpha},1-\frac{2}{\alpha},\frac{1}{1+sd^{-\alpha}}\right)\right)\label{eq:LT_K>1_n}\\
\mathcal{L}_{I_{-n}}(s,d)=&\exp\left(-\frac{2\pi}{\alpha}\left(1-T_n\right)\lambda_{b}s^{\frac{2}{\alpha}}B\left(\frac{2}{\alpha},1-\frac{2}{\alpha}\right)\right)\;.\label{eq:LT_K>1_k}
\end{align}
\normalsize{Substituting (\ref{eq:LT_K>1_n}) and (\ref{eq:LT_K>1_k}) into (\ref{eq:condi_CP_K}), we obtain  $q_{k,n,D_{0,0}}\left({\bf p},d\right)$ as follows:}
\begin{align}\label{eq:condi_CP_K_v2}
&q_{k,n,D_{0,0}}\left({\bf p},d\right)\nonumber\\
=&\exp\left(-\frac{2\pi}{\alpha}T_n\lambda_{b}d^{2}\left(2^{\frac{k\tau}{W}}-1\right)^{\frac{2}{\alpha}}B^{'}\left(\frac{2}{\alpha},1-\frac{2}{\alpha},2^{-\frac{k\tau}{W}}\right)\right) \exp\left(-\left(2^{\frac{k\tau}{W}}-1\right)d^{\alpha}\frac{N_{0}}{P}\right)\notag\\
&\times \exp\left(-\frac{2\pi}{\alpha}\left(1-T_n\right)\lambda_{b}d^{2}\left(2^{\frac{k\tau}{W}}-1\right)^{\frac{2}{\alpha}}B\left(\frac{2}{\alpha},1-\frac{2}{\alpha}\right)\right).
\end{align}

\normalsize{Now, we calculate $q_{K,n}\left({\bf p}\right)$ by first removing the condition of $q_{k,n,D_{0,0}}\left({\bf p},d\right)$ on $D_{0,0}=d$.  Note that we have the p.d.f. of $D_{0,0}$ as $f_{D_{0,0}}(d)=2\pi T_n\lambda_{b}d\exp\left(-\pi T_n\lambda_{b}d^{2}\right)$, as the BSs storing file $n$ form a homogeneous PPP with density $T_n\lambda_{b}$. Thus, we have:}
\begin{align}
&\int_{0}^{\infty}q_{k,n,D_{0,0}}\left({\bf p},d\right)f_{D_{0,0}}(d){\rm d}d\nonumber\\
=&2\pi T_n\lambda_{b}\int_{0}^{\infty}d\exp\left(-\pi T_n\lambda_{b}d^{2}\right)\exp\left(-\left(2^{\frac{k\tau}{W}}-1\right)d^{\alpha}\frac{N_{0}}{P}\right)\nonumber\\
&\times \exp\left(-\frac{2\pi}{\alpha}\left(1-T_n\right)\lambda_{b}d^{2}\left(2^{\frac{k\tau}{W}}-1\right)^{\frac{2}{\alpha}}B\left(\frac{2}{\alpha},1-\frac{2}{\alpha}\right)\right) \nonumber\\
&\times\exp\left(-\frac{2\pi}{\alpha}T_n\lambda_{b}d^{2}\left(2^{\frac{k\tau}{W}}-1\right)^{\frac{2}{\alpha}}B^{'}\left(\frac{2}{\alpha},1-\frac{2}{\alpha},2^{-\frac{k\tau}{W}}\right)\right){\rm d}d.\label{eq:CP_K_n}
\end{align}
\normalsize{Finally, by $q_K(\mathbf p)=\sum_{n\in \mathcal N}a_{n}\sum_{k=1}^K \Pr [K_{n,0}=k]\int_{0}^{\infty}q_{k,n,D_{0,0}}\left({\bf p},d\right)f_{D_{0,0}}(d){\rm d}d$, we can prove Theorem~\ref{Thm:generalKmulti}.}


\section*{Appendix G: Proof of Lemma~\ref{Lem:equivalence}}

Consider any feasible solution $\mathbf p$ to Problem~\ref{prob:opt-K-infty-p} satisfying \eqref{eqn:cache-constr-indiv} and \eqref{eqn:cache-constr-sum}. By \eqref{eqn:cache-constr-indiv} and \eqref{eqn:cache-constr-sum}, we know that $\mathbf T$ satisfies \eqref{eqn:cache-constr-indiv-t}. In addition, we have $\sum_{n\in \mathcal N}T_n=\sum_{n\in \mathcal N}\sum_{i\in \mathcal I_n}p_i=K\sum_{i\in \mathcal I}p_i=K$, i.e., $\mathbf T$ satisfies \eqref{eqn:cache-constr-sum-t}. Thus, $\mathbf T$ is a feasible solution to Problem~\ref{prob:opt-K-infty}. On the other hand, consider any feasible solution $\mathbf T$ to Problem~\ref{prob:opt-K-infty} satisfying \eqref{eqn:cache-constr-indiv-t} and \eqref{eqn:cache-constr-sum-t}. By the method in Fig. 1 of \cite{ICC15Giovanidis}, we can easily construct a feasible solution $\mathbf p$ to Problem~\ref{prob:opt-K-infty-p}. Therefore, we can show that Problem~\ref{prob:opt-K-infty-p} is equivalent to Problem~\ref{prob:opt-K-infty}. In other words, the optimal values of the two problems are the same.

\end{document}